\documentclass{llncs}
\usepackage{hyperref}
\usepackage{lipsum}
\usepackage{listings}

\usepackage[latin1]{inputenc}

\usepackage[english]{babel}

% ...

\usepackage{amsthm}

% ...

\usepackage{cite}

\usepackage[final]{microtype}

\usepackage{changebar}

\usepackage{etoolbox}

\usepackage{fixltx2e}

\usepackage{todonotes}

\usepackage{xargs}

\usepackage{xspace}

\usepackage{xstring}

\usepackage{boolexpr}

\usepackage[cmex10]{mathtools}
\usepackage{amssymb}
\usepackage{amsfonts}
\usepackage{mathrsfs}
\usepackage{latexsym}
\usepackage{textcomp}

\newcommand{\argemp}[2]
	{\if&#1&\else#2\fi}

\newcommand{\argdef}[2]
	{\if&#1&#2\else#1\fi}

\newcommand{\argint}[3]
	{\if&#2&\else#1#2#3\fi}

\newcommand{\argext}[3]
	{\if&#1&#3\else#1\if&#3&\else#2#3\fi\fi}

\newcommandx{\argsubsup}[3][2=, 3=]
	{\def\argsubscript{{#2}}\def\argsuperscript{{#3}}#1}

\newcommandx{\argind}[9][2=, 3=, 4=, 5=, 6=, 7=, 8=, 9=]
	{%
	\switch[#1=]%
		\case{0}#2%
		\case{1}#3%
		\case{2}#4%
		\case{3}#5%
		\case{4}#6%
		\case{5}#7%
		\case{6}#8%
		\case{7}#9%
		\otherwise\ensuremath{\clubsuit}%
	\endswitch%
	}

\newcommand{\arga}[1]
	{#1}
\newcommand{\argb}[2]
	{\argext{\arga{#1}}{,\allowbreak}{#2}}
\newcommand{\argc}[3]
	{\argext{\argb{#1}{#2}}{,\allowbreak}{#3}}
\newcommand{\argd}[4]
	{\argext{\argc{#1}{#2}{#3}}{,\allowbreak}{#4}}
\newcommand{\arge}[5]
	{\argext{\argd{#1}{#2}{#3}{#4}}{,\allowbreak}{#5}}
\newcommand{\argf}[6]
	{\argext{\arge{#1}{#2}{#3}{#4}{#5}}{,\allowbreak}{#6}}
\newcommand{\argg}[7]
	{\argext{\argf{#1}{#2}{#3}{#4}{#5}{#6}}{,\allowbreak}{#7}}
\newcommand{\argh}[8]
	{\argext{\argg{#1}{#2}{#3}{#4}{#5}{#6}{#7}}{,\allowbreak}{#8}}
\newcommand{\argi}[9]
	{\argext{\argh{#1}{#2}{#3}{#4}{#5}{#6}{#7}{#8}}{,\allowbreak}{#9}}

\newcommand{\txtfnt}[2][]
	{{%
	\IfStrEq{#1}{}
		{#2}
		{%
		\StrLeft{#1}{2}[\optbgn]%
		\StrGobbleLeft{#1}{2}[\optend]%
		\IfStrEqCase{\optbgn}
			{%
			{Rm}{\rmfamily\txtfnt[\optend]{#2}}%
			{Sf}{\sffamily\txtfnt[\optend]{#2}}%
			{Tt}{\ttfamily\txtfnt[\optend]{#2}}%
			{Up}{\upshape\txtfnt[\optend]{#2}}%
			{It}{\itshape\txtfnt[\optend]{#2}}%
			{Sl}{\slshape\txtfnt[\optend]{#2}}%
			{Sc}{\scshape\txtfnt[\optend]{#2}}%
			{Md}{\mdseries\txtfnt[\optend]{#2}}%
			{Bf}{\bfseries\txtfnt[\optend]{#2}}%
			{Em}{\emph{\txtfnt[\optend]{#2}}}%
			}
			[\ensuremath{\clubsuit}]%
		}%
	}}

\newcommand{\txtsub}[2][]
	{\argemp{#2}{\ensuremath{_{\text{\txtfnt[#1]{#2}}}}}}

\newcommand{\txtsup}[2][]
	{\argemp{#2}{\ensuremath{^{\text{\txtfnt[#1]{#2}}}}}}

\newcommandx{\txt}[4][1=, 3=, 4=]
	{{\txtfnt[#1]{#2}\ensuremath{\txtsub[#1]{#3}\txtsup[#1]{#4}}}}

\newcommandx{\txtarg}[5][1=, 3=, 4=]
	{{\txt[#1]{#2}[#3][#4]\argint{(}{#5}{)}}}

\newcommand{\txtstyname}{RmScMd}
\newcommand{\txtname}[1][]
	{\txt[\argdef{#1}{\txtstyname}]}
\newcommand{\txtargname}[1][]
	{\txtarg[\argdef{#1}{\txtstyname}]}

\newcommand{\txtstyabr}{Em}
\newcommand{\txtabr}[1][]
	{\txt[\argdef{#1}{\txtstyabr}]}

\newcommandx{\mthfnt}[3][1=, 2=0]
	{{%
	\IfStrEqCase{#1}
		{%
		{}%
			{#3}%
		{Name}%
			{%
			\IfStrEqCase{#2}
				{%
				{0}{\mathcal{#3}}%
				{1}{\mathscr{#3}}%
				{2}{\mathfrak{#3}}%
				{3}{\mathbb{#3}}%
				}
				[\ensuremath{\clubsuit}]%
			}%
		{Set}%
			{%
			\IfStrEqCase{#2}
				{%
				{0}{\mathrm{#3}}%
				{1}{\mathsf{#3}}%
				{2}{\mathbb{#3}}%
				{3}{\mathbf{#3}}%
				}
				[\ensuremath{\clubsuit}]%
			}%
		{Fun}%
			{%
			\IfStrEqCase{#2}
				{%
				{0}{\mathsf{#3}}%
				{1}{\mathrm{#3}}%
				}
				[\ensuremath{\clubsuit}]%
			}%
		{Rel}%
			{%
			\IfStrEqCase{#2}
				{%
				{0}{\mathit{#3}}%
				{1}{\mathtt{#3}}%
				}
				[\ensuremath{\clubsuit}]%
			}%
		{Sym}%
			{%
			\IfStrEqCase{#2}
				{%
				{0}{\mathtt{#3}}%
				{1}{\mathbf{#3}}%
				}
				[\ensuremath{\clubsuit}]%
			}%
		{Elm}%
			{\mathnormal{#3}}
		}
		[\ensuremath{\clubsuit}]%
	}}

\newcommand{\mthsub}[1]
	{\argemp{#1}{\ensuremath{_{\mathnormal{#1}}}}}

\newcommand{\mthsup}[1]
	{\argemp{#1}{\ensuremath{^{\mathnormal{#1}}}}}

\newcommandx{\mth}[5][1=, 2=0, 4=, 5=]
	{{\ensuremath{\mthfnt[#1][#2]{#3}\mthsub{#4}\mthsup{#5}}}}

\newcommandx{\mtharg}[6][1=, 2=0, 4=, 5=]
	{{\mth[#1][#2]{#3}[#4][#5]\ensuremath{\argint{(}{#6}{)}}}}

\newcommand{\mthempty}
	{\mth[][]}

\newcommand{\mthstyname}{0}
\newcommand{\mthname}[1][]
	{\mth[Name][\argdef{#1}{\mthstyname}]}

\newcommand{\mthstyset}{0}
\newcommand{\mthset}[1][]
	{\mth[Set][\argdef{#1}{\mthstyset}]}
\newcommand{\mthargset}[1][]
	{\mtharg[Set][\argdef{#1}{\mthstyset}]}

\newcommand{\mthstyfun}{0}
\newcommand{\mthfun}[1][]
	{\mth[Fun][\argdef{#1}{\mthstyfun}]}
\newcommand{\mthargfun}[1][]
	{\mtharg[Fun][\argdef{#1}{\mthstyfun}]}

\newcommand{\mthstyrel}{0}
\newcommand{\mthrel}[1][]
	{\mth[Rel][\argdef{#1}{\mthstyrel}]}

\newcommand{\mthstysym}{0}
\newcommand{\mthsym}[1][]
	{\mth[Sym][\argdef{#1}{\mthstysym}]}

\newcommand{\mthstyelm}{0}
\newcommand{\mthelm}[1][]
	{\mth[Elm][\argdef{#1}{\mthstyelm}]}

\newcommandx{\AName}[4][1=, 2=, 3=, 4=]{\mthname[#4]{A#3}[#1][#2]}
\newcommandx{\BName}[4][1=, 2=, 3=, 4=]{\mthname[#4]{B#3}[#1][#2]}
\newcommandx{\CName}[4][1=, 2=, 3=, 4=]{\mthname[#4]{C#3}[#1][#2]}
\newcommandx{\DName}[4][1=, 2=, 3=, 4=]{\mthname[#4]{D#3}[#1][#2]}
\newcommandx{\EName}[4][1=, 2=, 3=, 4=]{\mthname[#4]{E#3}[#1][#2]}
\newcommandx{\FName}[4][1=, 2=, 3=, 4=]{\mthname[#4]{F#3}[#1][#2]}
\newcommandx{\GName}[4][1=, 2=, 3=, 4=]{\mthname[#4]{G#3}[#1][#2]}
\newcommandx{\HName}[4][1=, 2=, 3=, 4=]{\mthname[#4]{H#3}[#1][#2]}
\newcommandx{\IName}[4][1=, 2=, 3=, 4=]{\mthname[#4]{I#3}[#1][#2]}
\newcommandx{\JName}[4][1=, 2=, 3=, 4=]{\mthname[#4]{J#3}[#1][#2]}
\newcommandx{\KName}[4][1=, 2=, 3=, 4=]{\mthname[#4]{K#3}[#1][#2]}
\newcommandx{\LName}[4][1=, 2=, 3=, 4=]{\mthname[#4]{L#3}[#1][#2]}
\newcommandx{\MName}[4][1=, 2=, 3=, 4=]{\mthname[#4]{M#3}[#1][#2]}
\newcommandx{\NName}[4][1=, 2=, 3=, 4=]{\mthname[#4]{N#3}[#1][#2]}
\newcommandx{\OName}[4][1=, 2=, 3=, 4=]{\mthname[#4]{O#3}[#1][#2]}
\newcommandx{\PName}[4][1=, 2=, 3=, 4=]{\mthname[#4]{P#3}[#1][#2]}
\newcommandx{\QName}[4][1=, 2=, 3=, 4=]{\mthname[#4]{Q#3}[#1][#2]}
\newcommandx{\RName}[4][1=, 2=, 3=, 4=]{\mthname[#4]{R#3}[#1][#2]}
\newcommandx{\SName}[4][1=, 2=, 3=, 4=]{\mthname[#4]{S#3}[#1][#2]}
\newcommandx{\TName}[4][1=, 2=, 3=, 4=]{\mthname[#4]{T#3}[#1][#2]}
\newcommandx{\UName}[4][1=, 2=, 3=, 4=]{\mthname[#4]{U#3}[#1][#2]}
\newcommandx{\VName}[4][1=, 2=, 3=, 4=]{\mthname[#4]{V#3}[#1][#2]}
\newcommandx{\WName}[4][1=, 2=, 3=, 4=]{\mthname[#4]{W#3}[#1][#2]}
\newcommandx{\XName}[4][1=, 2=, 3=, 4=]{\mthname[#4]{X#3}[#1][#2]}
\newcommandx{\YName}[4][1=, 2=, 3=, 4=]{\mthname[#4]{Y#3}[#1][#2]}
\newcommandx{\ZName}[4][1=, 2=, 3=, 4=]{\mthname[#4]{Z#3}[#1][#2]}

\newcommandx{\ASet}[4][1=, 2=, 3=, 4=]{\mthset[#4]{A#3}[#1][#2]}
\newcommandx{\BSet}[4][1=, 2=, 3=, 4=]{\mthset[#4]{B#3}[#1][#2]}
\newcommandx{\CSet}[4][1=, 2=, 3=, 4=]{\mthset[#4]{C#3}[#1][#2]}
\newcommandx{\DSet}[4][1=, 2=, 3=, 4=]{\mthset[#4]{D#3}[#1][#2]}
\newcommandx{\ESet}[4][1=, 2=, 3=, 4=]{\mthset[#4]{E#3}[#1][#2]}
\newcommandx{\FSet}[4][1=, 2=, 3=, 4=]{\mthset[#4]{F#3}[#1][#2]}
\newcommandx{\GSet}[4][1=, 2=, 3=, 4=]{\mthset[#4]{G#3}[#1][#2]}
\newcommandx{\HSet}[4][1=, 2=, 3=, 4=]{\mthset[#4]{H#3}[#1][#2]}
\newcommandx{\ISet}[4][1=, 2=, 3=, 4=]{\mthset[#4]{I#3}[#1][#2]}
\newcommandx{\JSet}[4][1=, 2=, 3=, 4=]{\mthset[#4]{J#3}[#1][#2]}
\newcommandx{\KSet}[4][1=, 2=, 3=, 4=]{\mthset[#4]{K#3}[#1][#2]}
\newcommandx{\LSet}[4][1=, 2=, 3=, 4=]{\mthset[#4]{L#3}[#1][#2]}
\newcommandx{\MSet}[4][1=, 2=, 3=, 4=]{\mthset[#4]{M#3}[#1][#2]}
\newcommandx{\NSet}[4][1=, 2=, 3=, 4=]{\mthset[#4]{N#3}[#1][#2]}
\newcommandx{\OSet}[4][1=, 2=, 3=, 4=]{\mthset[#4]{O#3}[#1][#2]}
\newcommandx{\PSet}[4][1=, 2=, 3=, 4=]{\mthset[#4]{P#3}[#1][#2]}
\newcommandx{\QSet}[4][1=, 2=, 3=, 4=]{\mthset[#4]{Q#3}[#1][#2]}
\newcommandx{\RSet}[4][1=, 2=, 3=, 4=]{\mthset[#4]{R#3}[#1][#2]}
\newcommandx{\SSet}[4][1=, 2=, 3=, 4=]{\mthset[#4]{S#3}[#1][#2]}
\newcommandx{\TSet}[4][1=, 2=, 3=, 4=]{\mthset[#4]{T#3}[#1][#2]}
\newcommandx{\USet}[4][1=, 2=, 3=, 4=]{\mthset[#4]{U#3}[#1][#2]}
\newcommandx{\VSet}[4][1=, 2=, 3=, 4=]{\mthset[#4]{V#3}[#1][#2]}
\newcommandx{\WSet}[4][1=, 2=, 3=, 4=]{\mthset[#4]{W#3}[#1][#2]}
\newcommandx{\XSet}[4][1=, 2=, 3=, 4=]{\mthset[#4]{X#3}[#1][#2]}
\newcommandx{\YSet}[4][1=, 2=, 3=, 4=]{\mthset[#4]{Y#3}[#1][#2]}
\newcommandx{\ZSet}[4][1=, 2=, 3=, 4=]{\mthset[#4]{Z#3}[#1][#2]}

\newcommandx{\aSet}[4][1=, 2=, 3=, 4=]{\mthset[#4]{a#3}[#1][#2]}
\newcommandx{\bSet}[4][1=, 2=, 3=, 4=]{\mthset[#4]{b#3}[#1][#2]}
\newcommandx{\cSet}[4][1=, 2=, 3=, 4=]{\mthset[#4]{c#3}[#1][#2]}
\newcommandx{\dSet}[4][1=, 2=, 3=, 4=]{\mthset[#4]{d#3}[#1][#2]}
\newcommandx{\eSet}[4][1=, 2=, 3=, 4=]{\mthset[#4]{e#3}[#1][#2]}
\newcommandx{\fSet}[4][1=, 2=, 3=, 4=]{\mthset[#4]{f#3}[#1][#2]}
\newcommandx{\gSet}[4][1=, 2=, 3=, 4=]{\mthset[#4]{g#3}[#1][#2]}
\newcommandx{\hSet}[4][1=, 2=, 3=, 4=]{\mthset[#4]{h#3}[#1][#2]}
\newcommandx{\iSet}[4][1=, 2=, 3=, 4=]{\mthset[#4]{i#3}[#1][#2]}
\newcommandx{\jSet}[4][1=, 2=, 3=, 4=]{\mthset[#4]{j#3}[#1][#2]}
\newcommandx{\kSet}[4][1=, 2=, 3=, 4=]{\mthset[#4]{k#3}[#1][#2]}
\newcommandx{\lSet}[4][1=, 2=, 3=, 4=]{\mthset[#4]{l#3}[#1][#2]}
\newcommandx{\mSet}[4][1=, 2=, 3=, 4=]{\mthset[#4]{m#3}[#1][#2]}
\newcommandx{\nSet}[4][1=, 2=, 3=, 4=]{\mthset[#4]{n#3}[#1][#2]}
\newcommandx{\oSet}[4][1=, 2=, 3=, 4=]{\mthset[#4]{o#3}[#1][#2]}
\newcommandx{\pSet}[4][1=, 2=, 3=, 4=]{\mthset[#4]{p#3}[#1][#2]}
\newcommandx{\qSet}[4][1=, 2=, 3=, 4=]{\mthset[#4]{q#3}[#1][#2]}
\newcommandx{\rSet}[4][1=, 2=, 3=, 4=]{\mthset[#4]{r#3}[#1][#2]}
\newcommandx{\sSet}[4][1=, 2=, 3=, 4=]{\mthset[#4]{s#3}[#1][#2]}
\newcommandx{\tSet}[4][1=, 2=, 3=, 4=]{\mthset[#4]{t#3}[#1][#2]}
\newcommandx{\uSet}[4][1=, 2=, 3=, 4=]{\mthset[#4]{u#3}[#1][#2]}
\newcommandx{\vSet}[4][1=, 2=, 3=, 4=]{\mthset[#4]{v#3}[#1][#2]}
\newcommandx{\wSet}[4][1=, 2=, 3=, 4=]{\mthset[#4]{w#3}[#1][#2]}
\newcommandx{\xSet}[4][1=, 2=, 3=, 4=]{\mthset[#4]{x#3}[#1][#2]}
\newcommandx{\ySet}[4][1=, 2=, 3=, 4=]{\mthset[#4]{y#3}[#1][#2]}
\newcommandx{\zSet}[4][1=, 2=, 3=, 4=]{\mthset[#4]{z#3}[#1][#2]}

\newcommandx{\AFun}[4][1=, 2=, 3=, 4=]{\mthfun[#4]{A#3}[#1][#2]}
\newcommandx{\BFun}[4][1=, 2=, 3=, 4=]{\mthfun[#4]{B#3}[#1][#2]}
\newcommandx{\CFun}[4][1=, 2=, 3=, 4=]{\mthfun[#4]{C#3}[#1][#2]}
\newcommandx{\DFun}[4][1=, 2=, 3=, 4=]{\mthfun[#4]{D#3}[#1][#2]}
\newcommandx{\EFun}[4][1=, 2=, 3=, 4=]{\mthfun[#4]{E#3}[#1][#2]}
\newcommandx{\FFun}[4][1=, 2=, 3=, 4=]{\mthfun[#4]{F#3}[#1][#2]}
\newcommandx{\GFun}[4][1=, 2=, 3=, 4=]{\mthfun[#4]{G#3}[#1][#2]}
\newcommandx{\HFun}[4][1=, 2=, 3=, 4=]{\mthfun[#4]{H#3}[#1][#2]}
\newcommandx{\IFun}[4][1=, 2=, 3=, 4=]{\mthfun[#4]{I#3}[#1][#2]}
\newcommandx{\JFun}[4][1=, 2=, 3=, 4=]{\mthfun[#4]{J#3}[#1][#2]}
\newcommandx{\KFun}[4][1=, 2=, 3=, 4=]{\mthfun[#4]{K#3}[#1][#2]}
\newcommandx{\LFun}[4][1=, 2=, 3=, 4=]{\mthfun[#4]{L#3}[#1][#2]}
\newcommandx{\MFun}[4][1=, 2=, 3=, 4=]{\mthfun[#4]{M#3}[#1][#2]}
\newcommandx{\NFun}[4][1=, 2=, 3=, 4=]{\mthfun[#4]{N#3}[#1][#2]}
\newcommandx{\OFun}[4][1=, 2=, 3=, 4=]{\mthfun[#4]{O#3}[#1][#2]}
\newcommandx{\PFun}[4][1=, 2=, 3=, 4=]{\mthfun[#4]{P#3}[#1][#2]}
\newcommandx{\QFun}[4][1=, 2=, 3=, 4=]{\mthfun[#4]{Q#3}[#1][#2]}
\newcommandx{\RFun}[4][1=, 2=, 3=, 4=]{\mthfun[#4]{R#3}[#1][#2]}
\newcommandx{\SFun}[4][1=, 2=, 3=, 4=]{\mthfun[#4]{S#3}[#1][#2]}
\newcommandx{\TFun}[4][1=, 2=, 3=, 4=]{\mthfun[#4]{T#3}[#1][#2]}
\newcommandx{\UFun}[4][1=, 2=, 3=, 4=]{\mthfun[#4]{U#3}[#1][#2]}
\newcommandx{\VFun}[4][1=, 2=, 3=, 4=]{\mthfun[#4]{V#3}[#1][#2]}
\newcommandx{\WFun}[4][1=, 2=, 3=, 4=]{\mthfun[#4]{W#3}[#1][#2]}
\newcommandx{\XFun}[4][1=, 2=, 3=, 4=]{\mthfun[#4]{X#3}[#1][#2]}
\newcommandx{\YFun}[4][1=, 2=, 3=, 4=]{\mthfun[#4]{Y#3}[#1][#2]}
\newcommandx{\ZFun}[4][1=, 2=, 3=, 4=]{\mthfun[#4]{Z#3}[#1][#2]}

\newcommandx{\aFun}[4][1=, 2=, 3=, 4=]{\mthfun[#4]{a#3}[#1][#2]}
\newcommandx{\bFun}[4][1=, 2=, 3=, 4=]{\mthfun[#4]{b#3}[#1][#2]}
\newcommandx{\cFun}[4][1=, 2=, 3=, 4=]{\mthfun[#4]{c#3}[#1][#2]}
\newcommandx{\dFun}[4][1=, 2=, 3=, 4=]{\mthfun[#4]{d#3}[#1][#2]}
\newcommandx{\eFun}[4][1=, 2=, 3=, 4=]{\mthfun[#4]{e#3}[#1][#2]}
\newcommandx{\fFun}[4][1=, 2=, 3=, 4=]{\mthfun[#4]{f#3}[#1][#2]}
\newcommandx{\gFun}[4][1=, 2=, 3=, 4=]{\mthfun[#4]{g#3}[#1][#2]}
\newcommandx{\hFun}[4][1=, 2=, 3=, 4=]{\mthfun[#4]{h#3}[#1][#2]}
\newcommandx{\iFun}[4][1=, 2=, 3=, 4=]{\mthfun[#4]{i#3}[#1][#2]}
\newcommandx{\jFun}[4][1=, 2=, 3=, 4=]{\mthfun[#4]{j#3}[#1][#2]}
\newcommandx{\kFun}[4][1=, 2=, 3=, 4=]{\mthfun[#4]{k#3}[#1][#2]}
\newcommandx{\lFun}[4][1=, 2=, 3=, 4=]{\mthfun[#4]{l#3}[#1][#2]}
\newcommandx{\mFun}[4][1=, 2=, 3=, 4=]{\mthfun[#4]{m#3}[#1][#2]}
\newcommandx{\nFun}[4][1=, 2=, 3=, 4=]{\mthfun[#4]{n#3}[#1][#2]}
\newcommandx{\oFun}[4][1=, 2=, 3=, 4=]{\mthfun[#4]{o#3}[#1][#2]}
\newcommandx{\pFun}[4][1=, 2=, 3=, 4=]{\mthfun[#4]{p#3}[#1][#2]}
\newcommandx{\qFun}[4][1=, 2=, 3=, 4=]{\mthfun[#4]{q#3}[#1][#2]}
\newcommandx{\rFun}[4][1=, 2=, 3=, 4=]{\mthfun[#4]{r#3}[#1][#2]}
\newcommandx{\sFun}[4][1=, 2=, 3=, 4=]{\mthfun[#4]{s#3}[#1][#2]}
\newcommandx{\tFun}[4][1=, 2=, 3=, 4=]{\mthfun[#4]{t#3}[#1][#2]}
\newcommandx{\uFun}[4][1=, 2=, 3=, 4=]{\mthfun[#4]{u#3}[#1][#2]}
\newcommandx{\vFun}[4][1=, 2=, 3=, 4=]{\mthfun[#4]{v#3}[#1][#2]}
\newcommandx{\wFun}[4][1=, 2=, 3=, 4=]{\mthfun[#4]{w#3}[#1][#2]}
\newcommandx{\xFun}[4][1=, 2=, 3=, 4=]{\mthfun[#4]{x#3}[#1][#2]}
\newcommandx{\yFun}[4][1=, 2=, 3=, 4=]{\mthfun[#4]{y#3}[#1][#2]}
\newcommandx{\zFun}[4][1=, 2=, 3=, 4=]{\mthfun[#4]{z#3}[#1][#2]}

\newcommandx{\ARel}[4][1=, 2=, 3=, 4=]{\mthrel[#4]{A#3}[#1][#2]}
\newcommandx{\BRel}[4][1=, 2=, 3=, 4=]{\mthrel[#4]{B#3}[#1][#2]}
\newcommandx{\CRel}[4][1=, 2=, 3=, 4=]{\mthrel[#4]{C#3}[#1][#2]}
\newcommandx{\DRel}[4][1=, 2=, 3=, 4=]{\mthrel[#4]{D#3}[#1][#2]}
\newcommandx{\ERel}[4][1=, 2=, 3=, 4=]{\mthrel[#4]{E#3}[#1][#2]}
\newcommandx{\FRel}[4][1=, 2=, 3=, 4=]{\mthrel[#4]{F#3}[#1][#2]}
\newcommandx{\GRel}[4][1=, 2=, 3=, 4=]{\mthrel[#4]{G#3}[#1][#2]}
\newcommandx{\HRel}[4][1=, 2=, 3=, 4=]{\mthrel[#4]{H#3}[#1][#2]}
\newcommandx{\IRel}[4][1=, 2=, 3=, 4=]{\mthrel[#4]{I#3}[#1][#2]}
\newcommandx{\JRel}[4][1=, 2=, 3=, 4=]{\mthrel[#4]{J#3}[#1][#2]}
\newcommandx{\KRel}[4][1=, 2=, 3=, 4=]{\mthrel[#4]{K#3}[#1][#2]}
\newcommandx{\LRel}[4][1=, 2=, 3=, 4=]{\mthrel[#4]{L#3}[#1][#2]}
\newcommandx{\MRel}[4][1=, 2=, 3=, 4=]{\mthrel[#4]{M#3}[#1][#2]}
\newcommandx{\NRel}[4][1=, 2=, 3=, 4=]{\mthrel[#4]{N#3}[#1][#2]}
\newcommandx{\ORel}[4][1=, 2=, 3=, 4=]{\mthrel[#4]{O#3}[#1][#2]}
\newcommandx{\PRel}[4][1=, 2=, 3=, 4=]{\mthrel[#4]{P#3}[#1][#2]}
\newcommandx{\QRel}[4][1=, 2=, 3=, 4=]{\mthrel[#4]{Q#3}[#1][#2]}
\newcommandx{\RRel}[4][1=, 2=, 3=, 4=]{\mthrel[#4]{R#3}[#1][#2]}
\newcommandx{\SRel}[4][1=, 2=, 3=, 4=]{\mthrel[#4]{S#3}[#1][#2]}
\newcommandx{\TRel}[4][1=, 2=, 3=, 4=]{\mthrel[#4]{T#3}[#1][#2]}
\newcommandx{\URel}[4][1=, 2=, 3=, 4=]{\mthrel[#4]{U#3}[#1][#2]}
\newcommandx{\VRel}[4][1=, 2=, 3=, 4=]{\mthrel[#4]{V#3}[#1][#2]}
\newcommandx{\WRel}[4][1=, 2=, 3=, 4=]{\mthrel[#4]{W#3}[#1][#2]}
\newcommandx{\XRel}[4][1=, 2=, 3=, 4=]{\mthrel[#4]{X#3}[#1][#2]}
\newcommandx{\YRel}[4][1=, 2=, 3=, 4=]{\mthrel[#4]{Y#3}[#1][#2]}
\newcommandx{\ZRel}[4][1=, 2=, 3=, 4=]{\mthrel[#4]{Z#3}[#1][#2]}

\newcommandx{\aRel}[4][1=, 2=, 3=, 4=]{\mthrel[#4]{a#3}[#1][#2]}
\newcommandx{\bRel}[4][1=, 2=, 3=, 4=]{\mthrel[#4]{b#3}[#1][#2]}
\newcommandx{\cRel}[4][1=, 2=, 3=, 4=]{\mthrel[#4]{c#3}[#1][#2]}
\newcommandx{\dRel}[4][1=, 2=, 3=, 4=]{\mthrel[#4]{d#3}[#1][#2]}
\newcommandx{\eRel}[4][1=, 2=, 3=, 4=]{\mthrel[#4]{e#3}[#1][#2]}
\newcommandx{\fRel}[4][1=, 2=, 3=, 4=]{\mthrel[#4]{f#3}[#1][#2]}
\newcommandx{\gRel}[4][1=, 2=, 3=, 4=]{\mthrel[#4]{g#3}[#1][#2]}
\newcommandx{\hRel}[4][1=, 2=, 3=, 4=]{\mthrel[#4]{h#3}[#1][#2]}
\newcommandx{\iRel}[4][1=, 2=, 3=, 4=]{\mthrel[#4]{i#3}[#1][#2]}
\newcommandx{\jRel}[4][1=, 2=, 3=, 4=]{\mthrel[#4]{j#3}[#1][#2]}
\newcommandx{\kRel}[4][1=, 2=, 3=, 4=]{\mthrel[#4]{k#3}[#1][#2]}
\newcommandx{\lRel}[4][1=, 2=, 3=, 4=]{\mthrel[#4]{l#3}[#1][#2]}
\newcommandx{\mRel}[4][1=, 2=, 3=, 4=]{\mthrel[#4]{m#3}[#1][#2]}
\newcommandx{\nRel}[4][1=, 2=, 3=, 4=]{\mthrel[#4]{n#3}[#1][#2]}
\newcommandx{\oRel}[4][1=, 2=, 3=, 4=]{\mthrel[#4]{o#3}[#1][#2]}
\newcommandx{\pRel}[4][1=, 2=, 3=, 4=]{\mthrel[#4]{p#3}[#1][#2]}
\newcommandx{\qRel}[4][1=, 2=, 3=, 4=]{\mthrel[#4]{q#3}[#1][#2]}
\newcommandx{\rRel}[4][1=, 2=, 3=, 4=]{\mthrel[#4]{r#3}[#1][#2]}
\newcommandx{\sRel}[4][1=, 2=, 3=, 4=]{\mthrel[#4]{s#3}[#1][#2]}
\newcommandx{\tRel}[4][1=, 2=, 3=, 4=]{\mthrel[#4]{t#3}[#1][#2]}
\newcommandx{\uRel}[4][1=, 2=, 3=, 4=]{\mthrel[#4]{u#3}[#1][#2]}
\newcommandx{\vRel}[4][1=, 2=, 3=, 4=]{\mthrel[#4]{v#3}[#1][#2]}
\newcommandx{\wRel}[4][1=, 2=, 3=, 4=]{\mthrel[#4]{w#3}[#1][#2]}
\newcommandx{\xRel}[4][1=, 2=, 3=, 4=]{\mthrel[#4]{x#3}[#1][#2]}
\newcommandx{\yRel}[4][1=, 2=, 3=, 4=]{\mthrel[#4]{y#3}[#1][#2]}
\newcommandx{\zRel}[4][1=, 2=, 3=, 4=]{\mthrel[#4]{z#3}[#1][#2]}

\newcommandx{\ASym}[4][1=, 2=, 3=, 4=]{\mthsym[#4]{A#3}[#1][#2]}
\newcommandx{\BSym}[4][1=, 2=, 3=, 4=]{\mthsym[#4]{B#3}[#1][#2]}
\newcommandx{\CSym}[4][1=, 2=, 3=, 4=]{\mthsym[#4]{C#3}[#1][#2]}
\newcommandx{\DSym}[4][1=, 2=, 3=, 4=]{\mthsym[#4]{D#3}[#1][#2]}
\newcommandx{\ESym}[4][1=, 2=, 3=, 4=]{\mthsym[#4]{E#3}[#1][#2]}
\newcommandx{\FSym}[4][1=, 2=, 3=, 4=]{\mthsym[#4]{F#3}[#1][#2]}
\newcommandx{\GSym}[4][1=, 2=, 3=, 4=]{\mthsym[#4]{G#3}[#1][#2]}
\newcommandx{\HSym}[4][1=, 2=, 3=, 4=]{\mthsym[#4]{H#3}[#1][#2]}
\newcommandx{\ISym}[4][1=, 2=, 3=, 4=]{\mthsym[#4]{I#3}[#1][#2]}
\newcommandx{\JSym}[4][1=, 2=, 3=, 4=]{\mthsym[#4]{J#3}[#1][#2]}
\newcommandx{\KSym}[4][1=, 2=, 3=, 4=]{\mthsym[#4]{K#3}[#1][#2]}
\newcommandx{\LSym}[4][1=, 2=, 3=, 4=]{\mthsym[#4]{L#3}[#1][#2]}
\newcommandx{\MSym}[4][1=, 2=, 3=, 4=]{\mthsym[#4]{M#3}[#1][#2]}
\newcommandx{\NSym}[4][1=, 2=, 3=, 4=]{\mthsym[#4]{N#3}[#1][#2]}
\newcommandx{\OSym}[4][1=, 2=, 3=, 4=]{\mthsym[#4]{O#3}[#1][#2]}
\newcommandx{\PSym}[4][1=, 2=, 3=, 4=]{\mthsym[#4]{P#3}[#1][#2]}
\newcommandx{\QSym}[4][1=, 2=, 3=, 4=]{\mthsym[#4]{Q#3}[#1][#2]}
\newcommandx{\RSym}[4][1=, 2=, 3=, 4=]{\mthsym[#4]{R#3}[#1][#2]}
\newcommandx{\SSym}[4][1=, 2=, 3=, 4=]{\mthsym[#4]{S#3}[#1][#2]}
\newcommandx{\TSym}[4][1=, 2=, 3=, 4=]{\mthsym[#4]{T#3}[#1][#2]}
\newcommandx{\USym}[4][1=, 2=, 3=, 4=]{\mthsym[#4]{U#3}[#1][#2]}
\newcommandx{\VSym}[4][1=, 2=, 3=, 4=]{\mthsym[#4]{V#3}[#1][#2]}
\newcommandx{\WSym}[4][1=, 2=, 3=, 4=]{\mthsym[#4]{W#3}[#1][#2]}
\newcommandx{\XSym}[4][1=, 2=, 3=, 4=]{\mthsym[#4]{X#3}[#1][#2]}
\newcommandx{\YSym}[4][1=, 2=, 3=, 4=]{\mthsym[#4]{Y#3}[#1][#2]}
\newcommandx{\ZSym}[4][1=, 2=, 3=, 4=]{\mthsym[#4]{Z#3}[#1][#2]}

\newcommandx{\aSym}[4][1=, 2=, 3=, 4=]{\mthsym[#4]{a#3}[#1][#2]}
\newcommandx{\bSym}[4][1=, 2=, 3=, 4=]{\mthsym[#4]{b#3}[#1][#2]}
\newcommandx{\cSym}[4][1=, 2=, 3=, 4=]{\mthsym[#4]{c#3}[#1][#2]}
\newcommandx{\dSym}[4][1=, 2=, 3=, 4=]{\mthsym[#4]{d#3}[#1][#2]}
\newcommandx{\eSym}[4][1=, 2=, 3=, 4=]{\mthsym[#4]{e#3}[#1][#2]}
\newcommandx{\fSym}[4][1=, 2=, 3=, 4=]{\mthsym[#4]{f#3}[#1][#2]}
\newcommandx{\gSym}[4][1=, 2=, 3=, 4=]{\mthsym[#4]{g#3}[#1][#2]}
\newcommandx{\hSym}[4][1=, 2=, 3=, 4=]{\mthsym[#4]{h#3}[#1][#2]}
\newcommandx{\iSym}[4][1=, 2=, 3=, 4=]{\mthsym[#4]{i#3}[#1][#2]}
\newcommandx{\jSym}[4][1=, 2=, 3=, 4=]{\mthsym[#4]{j#3}[#1][#2]}
\newcommandx{\kSym}[4][1=, 2=, 3=, 4=]{\mthsym[#4]{k#3}[#1][#2]}
\newcommandx{\lSym}[4][1=, 2=, 3=, 4=]{\mthsym[#4]{l#3}[#1][#2]}
\newcommandx{\mSym}[4][1=, 2=, 3=, 4=]{\mthsym[#4]{m#3}[#1][#2]}
\newcommandx{\nSym}[4][1=, 2=, 3=, 4=]{\mthsym[#4]{n#3}[#1][#2]}
\newcommandx{\oSym}[4][1=, 2=, 3=, 4=]{\mthsym[#4]{o#3}[#1][#2]}
\newcommandx{\pSym}[4][1=, 2=, 3=, 4=]{\mthsym[#4]{p#3}[#1][#2]}
\newcommandx{\qSym}[4][1=, 2=, 3=, 4=]{\mthsym[#4]{q#3}[#1][#2]}
\newcommandx{\rSym}[4][1=, 2=, 3=, 4=]{\mthsym[#4]{r#3}[#1][#2]}
\newcommandx{\sSym}[4][1=, 2=, 3=, 4=]{\mthsym[#4]{s#3}[#1][#2]}
\newcommandx{\tSym}[4][1=, 2=, 3=, 4=]{\mthsym[#4]{t#3}[#1][#2]}
\newcommandx{\uSym}[4][1=, 2=, 3=, 4=]{\mthsym[#4]{u#3}[#1][#2]}
\newcommandx{\vSym}[4][1=, 2=, 3=, 4=]{\mthsym[#4]{v#3}[#1][#2]}
\newcommandx{\wSym}[4][1=, 2=, 3=, 4=]{\mthsym[#4]{w#3}[#1][#2]}
\newcommandx{\xSym}[4][1=, 2=, 3=, 4=]{\mthsym[#4]{x#3}[#1][#2]}
\newcommandx{\ySym}[4][1=, 2=, 3=, 4=]{\mthsym[#4]{y#3}[#1][#2]}
\newcommandx{\zSym}[4][1=, 2=, 3=, 4=]{\mthsym[#4]{z#3}[#1][#2]}

\newcommandx{\AElm}[4][1=, 2=, 3=, 4=]{\mthelm[#4]{A#3}[#1][#2]}
\newcommandx{\BElm}[4][1=, 2=, 3=, 4=]{\mthelm[#4]{B#3}[#1][#2]}
\newcommandx{\CElm}[4][1=, 2=, 3=, 4=]{\mthelm[#4]{C#3}[#1][#2]}
\newcommandx{\DElm}[4][1=, 2=, 3=, 4=]{\mthelm[#4]{D#3}[#1][#2]}
\newcommandx{\EElm}[4][1=, 2=, 3=, 4=]{\mthelm[#4]{E#3}[#1][#2]}
\newcommandx{\FElm}[4][1=, 2=, 3=, 4=]{\mthelm[#4]{F#3}[#1][#2]}
\newcommandx{\GElm}[4][1=, 2=, 3=, 4=]{\mthelm[#4]{G#3}[#1][#2]}
\newcommandx{\HElm}[4][1=, 2=, 3=, 4=]{\mthelm[#4]{H#3}[#1][#2]}
\newcommandx{\IElm}[4][1=, 2=, 3=, 4=]{\mthelm[#4]{I#3}[#1][#2]}
\newcommandx{\JElm}[4][1=, 2=, 3=, 4=]{\mthelm[#4]{J#3}[#1][#2]}
\newcommandx{\KElm}[4][1=, 2=, 3=, 4=]{\mthelm[#4]{K#3}[#1][#2]}
\newcommandx{\LElm}[4][1=, 2=, 3=, 4=]{\mthelm[#4]{L#3}[#1][#2]}
\newcommandx{\MElm}[4][1=, 2=, 3=, 4=]{\mthelm[#4]{M#3}[#1][#2]}
\newcommandx{\NElm}[4][1=, 2=, 3=, 4=]{\mthelm[#4]{N#3}[#1][#2]}
\newcommandx{\OElm}[4][1=, 2=, 3=, 4=]{\mthelm[#4]{O#3}[#1][#2]}
\newcommandx{\PElm}[4][1=, 2=, 3=, 4=]{\mthelm[#4]{P#3}[#1][#2]}
\newcommandx{\QElm}[4][1=, 2=, 3=, 4=]{\mthelm[#4]{Q#3}[#1][#2]}
\newcommandx{\RElm}[4][1=, 2=, 3=, 4=]{\mthelm[#4]{R#3}[#1][#2]}
\newcommandx{\SElm}[4][1=, 2=, 3=, 4=]{\mthelm[#4]{S#3}[#1][#2]}
\newcommandx{\TElm}[4][1=, 2=, 3=, 4=]{\mthelm[#4]{T#3}[#1][#2]}
\newcommandx{\UElm}[4][1=, 2=, 3=, 4=]{\mthelm[#4]{U#3}[#1][#2]}
\newcommandx{\VElm}[4][1=, 2=, 3=, 4=]{\mthelm[#4]{V#3}[#1][#2]}
\newcommandx{\WElm}[4][1=, 2=, 3=, 4=]{\mthelm[#4]{W#3}[#1][#2]}
\newcommandx{\XElm}[4][1=, 2=, 3=, 4=]{\mthelm[#4]{X#3}[#1][#2]}
\newcommandx{\YElm}[4][1=, 2=, 3=, 4=]{\mthelm[#4]{Y#3}[#1][#2]}
\newcommandx{\ZElm}[4][1=, 2=, 3=, 4=]{\mthelm[#4]{Z#3}[#1][#2]}

\newcommandx{\aElm}[4][1=, 2=, 3=, 4=]{\mthelm[#4]{a#3}[#1][#2]}
\newcommandx{\bElm}[4][1=, 2=, 3=, 4=]{\mthelm[#4]{b#3}[#1][#2]}
\newcommandx{\cElm}[4][1=, 2=, 3=, 4=]{\mthelm[#4]{c#3}[#1][#2]}
\newcommandx{\dElm}[4][1=, 2=, 3=, 4=]{\mthelm[#4]{d#3}[#1][#2]}
\newcommandx{\eElm}[4][1=, 2=, 3=, 4=]{\mthelm[#4]{e#3}[#1][#2]}
\newcommandx{\fElm}[4][1=, 2=, 3=, 4=]{\mthelm[#4]{f#3}[#1][#2]}
\newcommandx{\gElm}[4][1=, 2=, 3=, 4=]{\mthelm[#4]{g#3}[#1][#2]}
\newcommandx{\hElm}[4][1=, 2=, 3=, 4=]{\mthelm[#4]{h#3}[#1][#2]}
\newcommandx{\iElm}[4][1=, 2=, 3=, 4=]{\mthelm[#4]{i#3}[#1][#2]}
\newcommandx{\jElm}[4][1=, 2=, 3=, 4=]{\mthelm[#4]{j#3}[#1][#2]}
\newcommandx{\kElm}[4][1=, 2=, 3=, 4=]{\mthelm[#4]{k#3}[#1][#2]}
\newcommandx{\lElm}[4][1=, 2=, 3=, 4=]{\mthelm[#4]{l#3}[#1][#2]}
\newcommandx{\mElm}[4][1=, 2=, 3=, 4=]{\mthelm[#4]{m#3}[#1][#2]}
\newcommandx{\nElm}[4][1=, 2=, 3=, 4=]{\mthelm[#4]{n#3}[#1][#2]}
\newcommandx{\oElm}[4][1=, 2=, 3=, 4=]{\mthelm[#4]{o#3}[#1][#2]}
\newcommandx{\pElm}[4][1=, 2=, 3=, 4=]{\mthelm[#4]{p#3}[#1][#2]}
\newcommandx{\qElm}[4][1=, 2=, 3=, 4=]{\mthelm[#4]{q#3}[#1][#2]}
\newcommandx{\rElm}[4][1=, 2=, 3=, 4=]{\mthelm[#4]{r#3}[#1][#2]}
\newcommandx{\sElm}[4][1=, 2=, 3=, 4=]{\mthelm[#4]{s#3}[#1][#2]}
\newcommandx{\tElm}[4][1=, 2=, 3=, 4=]{\mthelm[#4]{t#3}[#1][#2]}
\newcommandx{\uElm}[4][1=, 2=, 3=, 4=]{\mthelm[#4]{u#3}[#1][#2]}
\newcommandx{\vElm}[4][1=, 2=, 3=, 4=]{\mthelm[#4]{v#3}[#1][#2]}
\newcommandx{\wElm}[4][1=, 2=, 3=, 4=]{\mthelm[#4]{w#3}[#1][#2]}
\newcommandx{\xElm}[4][1=, 2=, 3=, 4=]{\mthelm[#4]{x#3}[#1][#2]}
\newcommandx{\yElm}[4][1=, 2=, 3=, 4=]{\mthelm[#4]{y#3}[#1][#2]}
\newcommandx{\zElm}[4][1=, 2=, 3=, 4=]{\mthelm[#4]{z#3}[#1][#2]}

\newcommand{\ie}
	{\txtabr{i.e.}\xspace}

\newcommand{\wrt}
	{\txtabr{w.r.t.}\xspace}

\newcommand{\resp}
	{\txtabr{resp.,}\xspace}

\newcommandx{\defeq}
	{\ensuremath{\triangleq}}

\newcommand{\fst}
	{\mthargfun{fst}}
\newcommand{\lst}
	{\mthargfun{lst}}

\newcommand{\tuple}[1]
	{\ensuremath{\argint{\langle}{#1}{\rangle}}}

\newcommand{\tupleb}[2]
	{\tuple{\argb{#1}{#2}}}
\newcommand{\tuplec}[3]
	{\tuple{\argc{#1}{#2}{#3}}}
\newcommand{\tupled}[4]
	{\tuple{\argd{#1}{#2}{#3}{#4}}}
\newcommand{\tuplee}[5]
	{\tuple{\arge{#1}{#2}{#3}{#4}{#5}}}
\newcommand{\tuplef}[6]
	{\tuple{\argf{#1}{#2}{#3}{#4}{#5}{#6}}}
\newcommand{\tupleg}[7]
	{\tuple{\argg{#1}{#2}{#3}{#4}{#5}{#6}{#7}}}
\newcommand{\tupleh}[8]
	{\tuple{\argh{#1}{#2}{#3}{#4}{#5}{#6}{#7}{#8}}}
\newcommand{\tuplei}[9]
	{\tuple{\argi{#1}{#2}{#3}{#4}{#5}{#6}{#7}{#8}{#9}}}

\newcommand{\tuplecx}[3]
	{%
	\def\defarga{#1}%
	\def\defargb{#2}%
	\def\defargc{#3}%
	\argsubsup{\tupleauxcx}%
	}
\newcommand{\tupledx}[4]
	{%
	\def\defarga{#1}%
	\def\defargb{#2}%
	\def\defargc{#3}%
	\def\defargd{#4}%
	\argsubsup{\tupleauxdx}%
	}
\newcommand{\tupleex}[5]
	{%
	\def\defarga{#1}%
	\def\defargb{#2}%
	\def\defargc{#3}%
	\def\defargd{#4}%
	\def\defarge{#5}%
	\argsubsup{\tupleauxex}%
	}
\newcommand{\tuplefx}[6]
	{%
	\def\defarga{#1}%
	\def\defargb{#2}%
	\def\defargc{#3}%
	\def\defargd{#4}%
	\def\defarge{#5}%
	\def\defargf{#6}%
	\argsubsup{\tupleauxfx}%
	}
\newcommand{\tuplegx}[7]
	{%
	\def\defarga{#1}%
	\def\defargb{#2}%
	\def\defargc{#3}%
	\def\defargd{#4}%
	\def\defarge{#5}%
	\def\defargf{#6}%
	\def\defargg{#7}%
	\argsubsup{\tupleauxgx}%
	}

\newcommandx{\tupleauxbx}[2][1=, 2=]
	{%
	\tupleb
		{\argdef{#1}{\defarga[\argsubscript][\argsuperscript]}}
		{\argdef{#2}{\defargb[\argsubscript][\argsuperscript]}}%
	}
\newcommandx{\tupleauxcx}[3][1=, 2=, 3=]
	{%
	\tuplec
		{\argdef{#1}{\defarga[\argsubscript][\argsuperscript]}}
		{\argdef{#2}{\defargb[\argsubscript][\argsuperscript]}}
		{\argdef{#3}{\defargc[\argsubscript][\argsuperscript]}}%
	}
\newcommandx{\tupleauxdx}[4][1=, 2=, 3=, 4=]
	{%
	\tupled
		{\argdef{#1}{\defarga[\argsubscript][\argsuperscript]}}
		{\argdef{#2}{\defargb[\argsubscript][\argsuperscript]}}
		{\argdef{#3}{\defargc[\argsubscript][\argsuperscript]}}
		{\argdef{#4}{\defargd[\argsubscript][\argsuperscript]}}%
	}
\newcommandx{\tupleauxex}[5][1=, 2=, 3=, 4=, 5=]
	{%
	\tuplee
		{\argdef{#1}{\defarga[\argsubscript][\argsuperscript]}}
		{\argdef{#2}{\defargb[\argsubscript][\argsuperscript]}}
		{\argdef{#3}{\defargc[\argsubscript][\argsuperscript]}}
		{\argdef{#4}{\defargd[\argsubscript][\argsuperscript]}}
		{\argdef{#5}{\defarge[\argsubscript][\argsuperscript]}}%
	}
\newcommandx{\tupleauxfx}[6][1=, 2=, 3=, 4=, 5=, 6=]
	{%
	\tuplef
		{\argdef{#1}{\defarga[\argsubscript][\argsuperscript]}}
		{\argdef{#2}{\defargb[\argsubscript][\argsuperscript]}}
		{\argdef{#3}{\defargc[\argsubscript][\argsuperscript]}}
		{\argdef{#4}{\defargd[\argsubscript][\argsuperscript]}}
		{\argdef{#5}{\defarge[\argsubscript][\argsuperscript]}}
		{\argdef{#6}{\defargf[\argsubscript][\argsuperscript]}}%
	}
\newcommandx{\tupleauxgx}[7][1=, 2=, 3=, 4=, 5=, 6=, 7=]
	{%
	\tupleg
		{\argdef{#1}{\defarga[\argsubscript][\argsuperscript]}}
		{\argdef{#2}{\defargb[\argsubscript][\argsuperscript]}}
		{\argdef{#3}{\defargc[\argsubscript][\argsuperscript]}}
		{\argdef{#4}{\defargd[\argsubscript][\argsuperscript]}}
		{\argdef{#5}{\defarge[\argsubscript][\argsuperscript]}}
		{\argdef{#6}{\defargf[\argsubscript][\argsuperscript]}}
		{\argdef{#7}{\defargg[\argsubscript][\argsuperscript]}}%
	}
\newcommandx{\tupleauxhx}[8][1=, 2=, 3=, 4=, 5=, 6=, 7=, 8=]
	{%
	\tupleh
		{\argdef{#1}{\defarga[\argsubscript][\argsuperscript]}}
		{\argdef{#2}{\defargb[\argsubscript][\argsuperscript]}}
		{\argdef{#3}{\defargc[\argsubscript][\argsuperscript]}}
		{\argdef{#4}{\defargd[\argsubscript][\argsuperscript]}}
		{\argdef{#5}{\defarge[\argsubscript][\argsuperscript]}}
		{\argdef{#6}{\defargf[\argsubscript][\argsuperscript]}}
		{\argdef{#7}{\defargg[\argsubscript][\argsuperscript]}}
		{\argdef{#8}{\defargh[\argsubscript][\argsuperscript]}}%
	}
\newcommandx{\tupleauxix}[9][1=, 2=, 3=, 4=, 5=, 6=, 7=, 8=, 9=]
	{%
	\tuplei
		{\argdef{#1}{\defarga[\argsubscript][\argsuperscript]}}
		{\argdef{#2}{\defargb[\argsubscript][\argsuperscript]}}
		{\argdef{#3}{\defargc[\argsubscript][\argsuperscript]}}
		{\argdef{#4}{\defargd[\argsubscript][\argsuperscript]}}
		{\argdef{#5}{\defarge[\argsubscript][\argsuperscript]}}
		{\argdef{#6}{\defargf[\argsubscript][\argsuperscript]}}
		{\argdef{#7}{\defargg[\argsubscript][\argsuperscript]}}
		{\argdef{#8}{\defargh[\argsubscript][\argsuperscript]}}
		{\argdef{#9}{\defargi[\argsubscript][\argsuperscript]}}%
	}

\newcommand{\set}[2]
	{\ensuremath{\argint{\{}{\argext{#1}{\allowbreak:\allowbreak}{#2}}{\}}}}

\newcommandx{\pto}[2][1=, 2=]
	{\ensuremath{\rightharpoonup}}

\newcommandx{\cto}[2][1=, 2=]
	{\:\mthempty{\to}[#1][#2]\:}
\newcommandx{\cpto}[2][1=, 2=]
	{\:\mthempty{\pto}[#1][#2]\:}

\newcommand{\SetN}
	{\mthset[2]{N}}

\newcommand{\numcc}[2]
	{\mthempty{[\argb{#1}{#2}]}}

\newcounter{explevautnot}

\providecommand{\autknd}{APT}

\providecommandx{\NFA}[5][1=, 2=, 3=, 4=, 5=]
	{\txtargname{Nfa#5{\small\argint{$[$}{#1}{$]$}}}[#2][#3]{#4}\xspace}

\providecommandx{\APT}[5][1=, 2=, 3=, 4=, 5=]
	{\txtargname{Apt#5{\small\argint{$[$}{#1}{$]$}}}[#2][#3]{#4}\xspace}

\providecommandx{\AlphSet}[3][1=, 2=, 3=]
	{\mthset{\Sigma#3}[#1][#2]}

\providecommandx{\alphSym}[3][1=, 2=, 3=]
	{\mthsym{\sigma#3}[#1][#2]}

\providecommandx{\alphElm}[3][1=, 2=, 3=]
	{\mthelm{\sigma#3}[#1][#2]}

\providecommandx{\AutStSet}[3][1=, 2=, 3=]
	{\mthset{Q#3}[#1][#2]}

\providecommandx{\autstSym}[3][1=, 2=, 3=]
	{\mthsym{q#3}[#1][#2]}

\providecommandx{\autstElm}[3][1=, 2=, 3=]
	{\mthelm{q#3}[#1][#2]}

\providecommandx{\trFun}[3][1=, 2=, 3=]
	{\mthfun{\delta#3}[#1][#2]}

\providecommandx{\BoolCom}{\BName[][+]}

\providecommand{\accCon}[1][]
	{%
	\IfStrEqCase{\argdef{#1}{\autknd}}
		{%
		{NFA}{\mthset{F}}%
		{APT}{\mthset{\beta}}%
		{NBA}{\mthset{\alpha}}%
		}
		[\ensuremath{\clubsuit}]%
	}

\providecommandx{\runElm}[4][1=, 2=, 3=, 4=]
	{\mthsym{\chi#4}[#1][#2]{#3}}

\providecommandx{\runSym}[4][1=, 2=, 3=, 4=]
	{\mthsym{\chi#4}[#1][#2]{#3}}

\providecommandx{\runStr}
	{\tupleb{\TSet}{\rFun}}

\providecommandx{\trkElm}[4][1=, 2=, 3=, 4=]
	{\mthelm{\eta#4}[#1][#2]{#3}}

\providecommandx{\trkSym}[4][1=, 2=, 3=, 4=]
	{\mthsym{\eta#4}[#1][#2]{#3}}

\providecommandx{\PthSet}[3][1=, 2=, 3=]
	{\mthset{Pth#3}[#1][#2]}

\providecommandx{\pthSym}[3][1=, 2=, 3=]
	{\mthsym{\varpi#3}[#1][#2]}

\providecommandx{\pthElm}[3][1=, 2=, 3=]
	{\mthelm{\varpi#3}[#1][#2]}

\providecommandx{\AutName}
	{\mthname{A}}

\providecommandx{\AutCls}[5][1=, 2=, 3=, 4=, 5=]
	{\mthset[#5]{Aut#4\text{\small\txtname{\argint{$[$}{#1}{$]$}}}}[#2][#3]}

\providecommand{\AutStr}[1][]
	{%
	\IfStrEqCase{\argdef{#1}{\autknd}}
		{%
		{NFA}{\tupled{\AlphSet}{\AutStSet}{\trFun}{\accCon[NFA]}}%
		{APT}{\tupled{\AlphSet}{\AutStSet}{\trFun}{\accCon[APT]}}%
		{NBA}{\tupled{\AlphSet}{\AutStSet}{\trFun}{\accCon[NBA]}}%
		}
		[\ensuremath{\clubsuit}]%
	}

\newcounter{explevgamnot}

\providecommand{\gamknd}{PG}

\providecommand{\frsentnam}{0}
\providecommand{\scnentnam}{1}
\providecommand{\finentnam}{\mthsym{F}}

\providecommandx{\PG}[5][1=, 2=, 3=, 4=, 5=]
	{\txtargname{Pg#5{\small\argint{$[$}{#1}{$]$}}}[#2][#3]{#4}\xspace}

\providecommandx{\EPG}[5][1=, 2=, 3=, 4=, 5=]
	{\txtargname{Epg#5{\small\argint{$[$}{#1}{$]$}}}[#2][#3]{#4}\xspace}

\providecommandx{\PlSet}[3][1=, 2=, 3=]
	{\mthset{Pl#3}[#1][#2]}

\providecommandx{\plSym}[3][1=, 2=, 3=]
	{\mthsym{p#3}[#1][#2]}

\providecommandx{\plElm}[3][1=, 2=, 3=]
	{\mthelm{p#3}[#1][#2]}

\providecommandx{\AgSet}[3][1=, 2=, 3=]
	{\mthset{Ag#3}[#1][#2]}

\providecommandx{\agSym}[3][1=, 2=, 3=]
	{\mthsym{a#3}[#1][#2]}

\providecommandx{\agElm}[3][1=, 2=, 3=]
	{\mthelm{a#3}[#1][#2]}

\providecommandx{\AcSet}[3][1=, 2=, 3=]
	{\mthset{Ac#3}[#1][#2]}

\providecommandx{\acSym}[3][1=, 2=, 3=]
	{\mthsym{c#3}[#1][#2]}

\providecommandx{\acElm}[3][1=, 2=, 3=]
	{\mthelm{c#3}[#1][#2]}

\providecommandx{\DcSet}[3][1=, 2=, 3=]
	{\mthset{Dc#3}[#1][#2]}

\providecommandx{\dcSym}[4][1=, 2=, 3=, 4=]
	{\mthargfun{\delta#4}[#1][#2]{#3}}

\providecommandx{\dcElm}[4][1=, 2=, 3=, 4=]
	{\mthargfun{\delta#4}[#1][#2]{#3}}

\providecommandx{\PsSet}[3][1=, 2=, 3=]
	{\mthset{Ps#3}[#1][#2]}

\providecommand{\FPsSet}[1][]
	{\PsSet[\frsentnam#1]}

\providecommand{\SPsSet}[1][]
	{\PsSet[\scnentnam#1]}

\providecommand{\FinPsSet}[1][]
	{\PsSet[\finentnam#1]}

\providecommandx{\psSym}[3][1=, 2=, 3=]
	{\mthsym{q#3}[#1][#2]}

\providecommandx{\psElm}[3][1=, 2=, 3=]
	{\mthelm{q#3}[#1][#2]}

\providecommandx{\StSet}[3][1=, 2=, 3=]
	{\mthset{St#3}[#1][#2]}

\providecommand{\FStSet}[1][]
	{\StSet[\frsentnam#1]}

\providecommand{\SStSet}[1][]
	{\StSet[\scnentnam#1]}

\providecommandx{\stSym}[3][1=, 2=, 3=]
	{\mthsym{s#3}[#1][#2]}

\providecommandx{\stElm}[3][1=, 2=, 3=]
	{\mthelm{s#3}[#1][#2]}

\providecommandx{\plFun}[4][1=, 2=, 3=, 4=]
	{\mthargfun{pl#4}[#1][#2]{#3}}

\providecommandx{\agFun}[4][1=, 2=, 3=, 4=]
	{\mthargfun{ag#4}[#1][#2]{#3}}

\providecommandx{\acFun}[4][1=, 2=, 3=, 4=]
	{\mthargfun{ac#4}[#1][#2]{#3}}

\providecommandx{\dcFun}[4][1=, 2=, 3=, 4=]
	{\mthargfun{dc#4}[#1][#2]{#3}}

\providecommandx{\mvFun}[4][1=, 2=, 3=, 4=]
	{\mthargfun{mv#4}[#1][#2]{#3}}

\providecommandx{\trFun}[4][1=, 2=, 3=, 4=]
	{\mthargfun{tr#4}[#1][#2]{#3}}

\providecommandx{\MvRel}[3][1=, 2=, 3=]
	{\mthrel{Mv#3}[#1][#2]}

\providecommandx{\mvSym}[3][1=, 2=, 3=]
	{\mthsym{m#3}[#1][#2]}

\providecommandx{\mvElm}[3][1=, 2=, 3=]
	{\mthelm{m#3}[#1][#2]}

\providecommand{\ArName}
	{\mthname{A}}

\providecommandx{\ArCls}[5][1=, 2=, 3=, 4=, 5=]
	{\mthset[#5]{Ar#4\text{\small\txtname{\argint{$[$}{#1}{$]$}}}}[#2][#3]}

\providecommand{\ArStr}[1][]
	{%
	\IfStrEqCase{\argdef{#1}{\gamknd}}
		{%
		{2PT}{\tupledx{\FPsSet}{\SPsSet}{\FinPsSet}{\MvRel}}%
		{2AT}{\tuplecx{\FStSet}{\SStSet}{\MvRel}}%
		{PG}{\tupledx{\FPsSet}{\SPsSet}{\MvRel}{\pFun}}%
		{EPG}{\tuplef{\FPsSet}{\SPsSet}{\VSet}{\ASet}{\MvRel}{\pFun}}%
		{MPC0}{\tupledx{\PlSet}{\AcSet}{\PsSet}{\trFun}}%
		{MAC0}{\tupledx{\AgSet}{\AcSet}{\StSet}{\trFun}}%
		{MPC1}{\tupleex{\PlSet}{\AcSet}{\PsSet}{\dcFun}{\trFun}}%
		{MAC1}{\tupleex{\AgSet}{\AcSet}{\StSet}{\dcFun}{\trFun}}%
		{MPC2}{\tuplefx{\PlSet}{\AcSet}{\PsSet}{\plFun}{\acFun}{\trFun}}%
		{MAC2}{\tuplefx{\AgSet}{\AcSet}{\StSet}{\agFun}{\acFun}{\trFun}}%
		{MPC3}{\tuplegx{\PlSet}{\AcSet}{\PsSet}{\plFun}{\acFun}{\dcFun}{\trFun}}%
		{MAC3}{\tuplegx{\AgSet}{\AcSet}{\StSet}{\agFun}{\acFun}{\dcFun}{\trFun}}%
		}
		[\ensuremath{\clubsuit}]%
	}

\providecommandx{\orFun}[3][1=, 2=, 3=]
	{\mthempty{\argint{\left\vert}{#3}{\right\vert}}[#1][#2]}

\providecommandx{\szFun}[3][1=, 2=, 3=]
	{\mthempty{\argint{\left\Vert}{#3}{\right\Vert}}[#1][#2]}

\providecommandx{\HstSet}[3][1=, 2=, 3=]
	{\mthset{Hst#3}[#1][#2]}

\providecommand{\FHstSet}[1][]
	{\HstSet[\frsentnam#1]}

\providecommand{\SHstSet}[1][]
	{\HstSet[\scnentnam#1]}

\providecommandx{\hstSym}[3][1=, 2=, 3=]
	{\mthsym{\rho#3}[#1][#2]}

\providecommandx{\hstElm}[3][1=, 2=, 3=]
	{\mthelm{\rho#3}[#1][#2]}

\providecommandx{\StrSet}[3][1=, 2=, 3=]
	{\mthset{Str#3}[#1][#2]}

\providecommandx{\strSym}[4][1=, 2=, 3=, 4=]
	{\mthargfun{str#4}[#1][#2]{#3}}

\providecommandx{\strElm}[4][1=, 2=, 3=, 4=]
	{\mthargfun{str#4}[#1][#2]{#3}}

\providecommand{\fstrElm}[1][]
	{\strElm[\frsentnam#1]}

\providecommand{\sstrElm}[1][]
	{\strElm[\scnentnam#1]}

\providecommandx{\PrfSet}[3][1=, 2=, 3=]
	{\mthset{Prf#3}[#1][#2]}

\providecommandx{\prfSym}[4][1=, 2=, 3=, 4=]
	{\mthargfun{\xi#4}[#1][#2]{#3}}

\providecommandx{\prfElm}[4][1=, 2=, 3=, 4=]
	{\mthargfun{\xi#4}[#1][#2]{#3}}

\providecommandx{\playFun}[4][1=, 2=, 3=, 4=]
	{\mthargfun{play#4}[#1][#2]{#3}}

\providecommandx{\playElm}[4][1=, 2=, 3=, 4=]
	{\mthelm{\pi#4}[#1][#2]{#3}}

\providecommandx{\playSym}[4][1=, 2=, 3=, 4=]
	{\mthsym{\pi#4}[#1][#2]{#3}}

\providecommandx{\playSet}[4][1=, 2=, 3=, 4=]
	{\mthset{Ply#4}[#1][#2]{#3}}

\providecommandx{\PfSet}[3][1=, 2=, 3=]
	{\mthset{Pf#3}[#1][#2]}

\providecommandx{\pfFun}[4][1=, 2=, 3=, 4=]
	{\mthargfun{pf#4}[#1][#2]{#3}}

\providecommandx{\PfArName}[2][1=, 2=]
	{\widehat{\ArName[#1][#2]}}

\providecommandx{\PfArCls}[5][1=, 2=, 3=, 4=, 5=]
	{\widehat{\ArCls[#1][#2][#3][#4][#5]}}

\providecommandx{\ClSet}[3][1=, 2=, 3=]
	{\mthset{Cl#3}[#1][#2]}

\providecommandx{\clFun}[4][1=, 2=, 3=, 4=]
	{\mthargfun{cl#4}[#1][#2]{#3}}

\providecommandx{\ClArName}[2][1=, 2=]
	{\widetilde{\ArName[#1][#2]}}

\providecommandx{\ClArCls}[5][1=, 2=, 3=, 4=, 5=]
	{\widetilde{\ArCls[#1][#2][#3][#4][#5]}}

\providecommandx{\WgSet}[3][1=, 2=, 3=]
	{\mthset{Wg#3}[#1][#2]}

\providecommandx{\wgFun}[4][1=, 2=, 3=, 4=]
	{\mthargfun{wg#4}[#1][#2]{#3}}

\providecommandx{\WgArName}[2][1=, 2=]
	{\widetilde{\ArName[#1][#2]}}

\providecommandx{\WgArCls}[5][1=, 2=, 3=, 4=, 5=]
	{\widetilde{\ArCls[#1][#2][#3][#4][#5]}}

\providecommandx{\ClWgArName}[2][1=, 2=]
	{\overline{\ArName[#1][#2]}}

\providecommandx{\ClWgArCls}[5][1=, 2=, 3=, 4=, 5=]
	{\overline{\ArCls[#1][#2][#3][#4][#5]}}

% ...

\providecommandx{\WnSet}[3][1=, 2=, 3=]
	{\mthset{Wn#3}[#1][#2]}

\providecommandx{\wnSym}[3][1=, 2=, 3=]
	{\mthsym{w#3}[#1][#2]}

\providecommandx{\wnElm}[3][1=, 2=, 3=]
	{\mthelm{w#3}[#1][#2]}

% ...

\providecommand{\GmName}
	{\mthname{G}}

\providecommandx{\GmCls}[5][1=, 2=, 3=, 4=, 5=]
	{\mthset[#5]{Gm#4\text{\small\txtname{\argint{$[$}{#1}{$]$}}}}[#2][#3]}

\newcommandx{\GmStr}[2][1=, 2=]
	{%
	\StrLeft{\argdef{#1}{\gamknd}}{2}[\optgmknd]%
	\def\defpselm{#2}%
	\argsubsup{\GmStrAux}%
	}
\newcommandx{\GmStrAux}[2][1=, 2=]
	{%
	\IfStrEqCase{\optgmknd}
		{%
		{2P}
			{%
			\tuplec
				{\argdef{#1}{\ArName[\argsubscript][\argsuperscript]}}
				{\argdef{#2}{\WnSet[\argsubscript][\argsuperscript]}}
				{\argdef{\defpselm}{\psElm[0\argsubscript][\argsuperscript]}}%
			}%
		{2A}
			{%
			\tuplec
				{\argdef{#1}{\PfArName[\argsubscript][\argsuperscript]}}
				{\argdef{#2}{\WnSet[\argsubscript][\argsuperscript]}}
				{\argdef{\defpselm}{\stElm[0\argsubscript][\argsuperscript]}}%
			}%
		{MP}
			{%
			\tuplec
				{\argdef{#1}{\PfArName[\argsubscript][\argsuperscript]}}
				{\argdef{#2}{...}}
				{\argdef{\defpselm}{\psElm[0\argsubscript][\argsuperscript]}}%
			}%
		{MA}
			{%
			\tuplec
				{\argdef{#1}{\PfArName[\argsubscript][\argsuperscript]}}
				{\argdef{#2}{...}}
				{\argdef{\defpselm}{\stElm[0\argsubscript][\argsuperscript]}}%
			}%
		}
		[\ensuremath{\clubsuit}]%
	}

\newcommandx{\CTL}[5][1=, 2=, 3=, 4=, 5=]
	{\txtargname{CTL#5{\small\argint{$[$}{#1}{$]$}}}[#2][#3]{#4}\xspace}

\newcommandx{\CTLP}[5][1=, 2=, 3=, 4=, 5=]
	{\txtargname{CTL$^{+}$#5{\small\argint{$[$}{#1}{$]$}}}[#2][#3]{#4}\xspace}

\newcommandx{\CTLS}[5][1=, 2=, 3=, 4=, 5=]
	{\txtargname{CTL$^{\star}$#5{\small\argint{$[$}{#1}{$]$}}}[#2][#3]{#4}\xspace}

\newcommandx{\LogTime}[4][1=, 2=, 3=, 4=]
	{\txtargname{LogTime#4}[#2][#3]{#1}\xspace}
\newcommandx{\LogTimeH}[4][1=, 2=, 3=, 4=]
	{\LogTime[#1][#2][#3][#4]-\HComp}
\newcommandx{\LogTimeE}[4][1=, 2=, 3=, 4=]
	{\LogTime[#1][#2][#3][#4]-\EComp}
\newcommandx{\LogTimeC}[4][1=, 2=, 3=, 4=]
	{\LogTime[#1][#2][#3][#4]-\CComp}

\newcommand{\NLogTime}
	{\txtname{N}\LogTime}
\newcommandx{\NLogTimeH}[4][1=, 2=, 3=, 4=]
	{\NLogTime[#1][#2][#3][#4]-\HComp}
\newcommandx{\NLogTimeE}[4][1=, 2=, 3=, 4=]
	{\NLogTime[#1][#2][#3][#4]-\EComp}
\newcommandx{\NLogTimeC}[4][1=, 2=, 3=, 4=]
	{\NLogTime[#1][#2][#3][#4]-\CComp}

\newcommand{\CoNLogTime}
	{\txtname{Co}\NLogTime}
\newcommandx{\CoNLogTimeH}[4][1=, 2=, 3=, 4=]
	{\CoNLogTime[#1][#2][#3][#4]-\HComp}
\newcommandx{\CoNLogTimeE}[4][1=, 2=, 3=, 4=]
	{\CoNLogTime[#1][#2][#3][#4]-\EComp}
\newcommandx{\CoNLogTimeC}[4][1=, 2=, 3=, 4=]
	{\CoNLogTime[#1][#2][#3][#4]-\CComp}

\newcommandx{\ALogTime}
	{\txtname{A}\LogTime}
\newcommandx{\ALogTimeH}[4][1=, 2=, 3=, 4=]
	{\ALogTime[#1][#2][#3][#4]-\HComp}
\newcommandx{\ALogTimeE}[4][1=, 2=, 3=, 4=]
	{\ALogTime[#1][#2][#3][#4]-\EComp}
\newcommandx{\ALogTimeC}[4][1=, 2=, 3=, 4=]
	{\ALogTime[#1][#2][#3][#4]-\CComp}

\newcommandx{\LogSpace}[4][1=, 2=, 3=, 4=]
	{\txtargname{LogSpace#4}[#2][#3]{#1}\xspace}
\newcommandx{\LogSpaceH}[4][1=, 2=, 3=, 4=]
	{\LogSpace[#1][#2][#3][#4]-\HComp}
\newcommandx{\LogSpaceE}[4][1=, 2=, 3=, 4=]
	{\LogSpace[#1][#2][#3][#4]-\EComp}
\newcommandx{\LogSpaceC}[4][1=, 2=, 3=, 4=]
	{\LogSpace[#1][#2][#3][#4]-\CComp}

\newcommandx{\NLogSpace}
	{\txtname{N}\LogSpace}
\newcommandx{\NLogSpaceH}[4][1=, 2=, 3=, 4=]
	{\NLogSpace[#1][#2][#3][#4]-\HComp}
\newcommandx{\NLogSpaceE}[4][1=, 2=, 3=, 4=]
	{\NLogSpace[#1][#2][#3][#4]-\EComp}
\newcommandx{\NLogSpaceC}[4][1=, 2=, 3=, 4=]
	{\NLogSpace[#1][#2][#3][#4]-\CComp}

\newcommandx{\CoNLogSpace}
	{\txtname{Co}\NLogSpace}
\newcommandx{\CoNLogSpaceH}[4][1=, 2=, 3=, 4=]
	{\CoNLogSpace[#1][#2][#3][#4]-\HComp}
\newcommandx{\CoNLogSpaceE}[4][1=, 2=, 3=, 4=]
	{\CoNLogSpace[#1][#2][#3][#4]-\EComp}
\newcommandx{\CoNLogSpaceC}[4][1=, 2=, 3=, 4=]
	{\CoNLogSpace[#1][#2][#3][#4]-\CComp}

\newcommandx{\ALogSpace}
	{\txtname{A}\LogSpace}
\newcommandx{\ALogSpaceH}[4][1=, 2=, 3=, 4=]
	{\ALogSpace[#1][#2][#3][#4]-\HComp}
\newcommandx{\ALogSpaceE}[4][1=, 2=, 3=, 4=]
	{\ALogSpace[#1][#2][#3][#4]-\EComp}
\newcommandx{\ALogSpaceC}[4][1=, 2=, 3=, 4=]
	{\ALogSpace[#1][#2][#3][#4]-\CComp}

\newcommandx{\PTime}[4][1=, 2=, 3=, 4=]
	{\txtargname{PTime#4}[#2][#3]{#1}\xspace}
\newcommandx{\PTimeH}[4][1=, 2=, 3=, 4=]
	{\PTime[#1][#2][#3][#4]-\HComp}
\newcommandx{\PTimeE}[4][1=, 2=, 3=, 4=]
	{\PTime[#1][#2][#3][#4]-\EComp}
\newcommandx{\PTimeC}[4][1=, 2=, 3=, 4=]
	{\PTime[#1][#2][#3][#4]-\CComp}

\newcommandx{\UPTime}
	{\txtname{U}\PTime}
\newcommandx{\UPTimeH}[4][1=, 2=, 3=, 4=]
	{\UPTime[#1][#2][#3][#4]-\HComp}
\newcommandx{\UPTimeE}[4][1=, 2=, 3=, 4=]
	{\UPTime[#1][#2][#3][#4]-\EComp}
\newcommandx{\UPTimeC}[4][1=, 2=, 3=, 4=]
	{\UPTime[#1][#2][#3][#4]-\CComp}

\newcommandx{\CoUPTime}
	{\txtname{Co}\UPTime}
\newcommandx{\CoUPTimeH}[4][1=, 2=, 3=, 4=]
	{\CoUPTime[#1][#2][#3][#4]-\HComp}
\newcommandx{\CoUPTimeE}[4][1=, 2=, 3=, 4=]
	{\CoUPTime[#1][#2][#3][#4]-\EComp}
\newcommandx{\CoUPTimeC}[4][1=, 2=, 3=, 4=]
	{\CoUPTime[#1][#2][#3][#4]-\CComp}

\newcommandx{\NPTime}
	{\txtname{N}\PTime}
\newcommandx{\NPTimeH}[4][1=, 2=, 3=, 4=]
	{\NPTime[#1][#2][#3][#4]-\HComp}
\newcommandx{\NPTimeE}[4][1=, 2=, 3=, 4=]
	{\NPTime[#1][#2][#3][#4]-\EComp}
\newcommandx{\NPTimeC}[4][1=, 2=, 3=, 4=]
	{\NPTime[#1][#2][#3][#4]-\CComp}

\newcommandx{\CoNPTime}
	{\txtname{Co}\NPTime}
\newcommandx{\CoNPTimeH}[4][1=, 2=, 3=, 4=]
	{\CoNPTime[#1][#2][#3][#4]-\HComp}
\newcommandx{\CoNPTimeE}[4][1=, 2=, 3=, 4=]
	{\CoNPTime[#1][#2][#3][#4]-\EComp}
\newcommandx{\CoNPTimeC}[4][1=, 2=, 3=, 4=]
	{\CoNPTime[#1][#2][#3][#4]-\CComp}

\newcommandx{\APTime}
	{\txtname{A}\PTime}
\newcommandx{\APTimeH}[4][1=, 2=, 3=, 4=]
	{\APTime[#1][#2][#3][#4]-\HComp}
\newcommandx{\APTimeE}[4][1=, 2=, 3=, 4=]
	{\APTime[#1][#2][#3][#4]-\EComp}
\newcommandx{\APTimeC}[4][1=, 2=, 3=, 4=]
	{\APTime[#1][#2][#3][#4]-\CComp}

\newcommandx{\PSpace}[4][1=, 2=, 3=, 4=]
	{\txtargname{PSpace#4}[#2][#3]{#1}\xspace}
\newcommandx{\PSpaceH}[4][1=, 2=, 3=, 4=]
	{\PSpace[#1][#2][#3][#4]-\HComp}
\newcommandx{\PSpaceE}[4][1=, 2=, 3=, 4=]
	{\PSpace[#1][#2][#3][#4]-\EComp}
\newcommandx{\PSpaceC}[4][1=, 2=, 3=, 4=]
	{\PSpace[#1][#2][#3][#4]-\CComp}

\newcommandx{\NPSpace}
	{\txtname{N}\PSpace}
\newcommandx{\NPSpaceH}[4][1=, 2=, 3=, 4=]
	{\NPSpace[#1][#2][#3][#4]-\HComp}
\newcommandx{\NPSpaceE}[4][1=, 2=, 3=, 4=]
	{\NPSpace[#1][#2][#3][#4]-\EComp}
\newcommandx{\NPSpaceC}[4][1=, 2=, 3=, 4=]
	{\NPSpace[#1][#2][#3][#4]-\CComp}

\newcommandx{\CoNPSpace}
	{\txtname{Co}\NPSpace}
\newcommandx{\CoNPSpaceH}[4][1=, 2=, 3=, 4=]
	{\CoNPSpace[#1][#2][#3][#4]-\HComp}
\newcommandx{\CoNPSpaceE}[4][1=, 2=, 3=, 4=]
	{\CoNPSpace[#1][#2][#3][#4]-\EComp}
\newcommandx{\CoNPSpaceC}[4][1=, 2=, 3=, 4=]
	{\CoNPSpace[#1][#2][#3][#4]-\CComp}

\newcommandx{\APSpace}
	{\txtname{A}\PSpace}
\newcommandx{\APSpaceH}[4][1=, 2=, 3=, 4=]
	{\APSpace[#1][#2][#3][#4]-\HComp}
\newcommandx{\APSpaceE}[4][1=, 2=, 3=, 4=]
	{\APSpace[#1][#2][#3][#4]-\EComp}
\newcommandx{\APSpaceC}[4][1=, 2=, 3=, 4=]
	{\APSpace[#1][#2][#3][#4]-\CComp}

\newcommandx{\ExpTime}[4][1=, 2=, 3=, 4=]
	{\txtargname{ExpTime#4}[#2][#3]{#1}\xspace}
\newcommandx{\ExpTimeH}[4][1=, 2=, 3=, 4=]
	{\ExpTime[#1][#2][#3][#4]-\HComp}
\newcommandx{\ExpTimeE}[4][1=, 2=, 3=, 4=]
	{\ExpTime[#1][#2][#3][#4]-\EComp}
\newcommandx{\ExpTimeC}[4][1=, 2=, 3=, 4=]
	{\ExpTime[#1][#2][#3][#4]-\CComp}

\newcommandx{\NExpTime}
	{\txtname{N}\ExpTime}
\newcommandx{\NExpTimeH}[4][1=, 2=, 3=, 4=]
	{\NExpTime[#1][#2][#3][#4]-\HComp}
\newcommandx{\NExpTimeE}[4][1=, 2=, 3=, 4=]
	{\NExpTime[#1][#2][#3][#4]-\EComp}
\newcommandx{\NExpTimeC}[4][1=, 2=, 3=, 4=]
	{\NExpTime[#1][#2][#3][#4]-\CComp}

\newcommandx{\CoNExpTime}
	{\txtname{Co}\NExpTime}
\newcommandx{\CoNExpTimeH}[4][1=, 2=, 3=, 4=]
	{\CoNExpTime[#1][#2][#3][#4]-\HComp}
\newcommandx{\CoNExpTimeE}[4][1=, 2=, 3=, 4=]
	{\CoNExpTime[#1][#2][#3][#4]-\EComp}
\newcommandx{\CoNExpTimeC}[4][1=, 2=, 3=, 4=]
	{\CoNExpTime[#1][#2][#3][#4]-\CComp}

\newcommandx{\AExpTime}
	{\txtname{A}\ExpTime}
\newcommandx{\AExpTimeH}[4][1=, 2=, 3=, 4=]
	{\AExpTime[#1][#2][#3][#4]-\HComp}
\newcommandx{\AExpTimeE}[4][1=, 2=, 3=, 4=]
	{\AExpTime[#1][#2][#3][#4]-\EComp}
\newcommandx{\AExpTimeC}[4][1=, 2=, 3=, 4=]
	{\AExpTime[#1][#2][#3][#4]-\CComp}

\newcommandx{\ExpSpace}[4][1=, 2=, 3=, 4=]
	{\txtargname{ExpSpace#4}[#2][#3]{#1}\xspace}
\newcommandx{\ExpSpaceH}[4][1=, 2=, 3=, 4=]
	{\ExpSpace[#1][#2][#3][#4]-\HComp}
\newcommandx{\ExpSpaceE}[4][1=, 2=, 3=, 4=]
	{\ExpSpace[#1][#2][#3][#4]-\EComp}
\newcommandx{\ExpSpaceC}[4][1=, 2=, 3=, 4=]
	{\ExpSpace[#1][#2][#3][#4]-\CComp}

\newcommandx{\NExpSpace}
	{\txtname{N}\ExpSpace}
\newcommandx{\NExpSpaceH}[4][1=, 2=, 3=, 4=]
	{\NExpSpace[#1][#2][#3][#4]-\HComp}
\newcommandx{\NExpSpaceE}[4][1=, 2=, 3=, 4=]
	{\NExpSpace[#1][#2][#3][#4]-\EComp}
\newcommandx{\NExpSpaceC}[4][1=, 2=, 3=, 4=]
	{\NExpSpace[#1][#2][#3][#4]-\CComp}

\newcommandx{\CoNExpSpace}
	{\txtname{Co}\NExpSpace}
\newcommandx{\CoNExpSpaceH}[4][1=, 2=, 3=, 4=]
	{\CoNExpSpace[#1][#2][#3][#4]-\HComp}
\newcommandx{\CoNExpSpaceE}[4][1=, 2=, 3=, 4=]
	{\CoNExpSpace[#1][#2][#3][#4]-\EComp}
\newcommandx{\CoNExpSpaceC}[4][1=, 2=, 3=, 4=]
	{\CoNExpSpace[#1][#2][#3][#4]-\CComp}

\newcommandx{\AExpSpace}
	{\txtname{A}\ExpSpace}
\newcommandx{\AExpSpaceH}[4][1=, 2=, 3=, 4=]
	{\AExpSpace[#1][#2][#3][#4]-\HComp}
\newcommandx{\AExpSpaceE}[4][1=, 2=, 3=, 4=]
	{\AExpSpace[#1][#2][#3][#4]-\EComp}
\newcommandx{\AExpSpaceC}[4][1=, 2=, 3=, 4=]
	{\AExpSpace[#1][#2][#3][#4]-\CComp}

\newcommandx{\NonEleTime}[4][1=, 2=, 3=, 4=]
	{\txtargname{NonElementaryTime#4}[#2][#3]{#1}\xspace}
\newcommandx{\NonEleTimeH}[4][1=, 2=, 3=, 4=]
	{\NonEleTime[#1][#2][#3][#4]-\HComp}
\newcommandx{\NonEleTimeE}[4][1=, 2=, 3=, 4=]
	{\NonEleTime[#1][#2][#3][#4]-\EComp}
\newcommandx{\NonEleTimeC}[4][1=, 2=, 3=, 4=]
	{\NonEleTime[#1][#2][#3][#4]-\CComp}

\newcommandx{\NonEleSpace}[4][1=, 2=, 3=, 4=]
	{\txtargname{NonElementarySpace#4}[#2][#3]{#1}\xspace}
\newcommandx{\NonEleSpaceH}[4][1=, 2=, 3=, 4=]
	{\NonEleSpace[#1][#2][#3][#4]-\HComp}
\newcommandx{\NonEleSpaceE}[4][1=, 2=, 3=, 4=]
	{\NonEleSpace[#1][#2][#3][#4]-\EComp}
\newcommandx{\NonEleSpaceC}[4][1=, 2=, 3=, 4=]
	{\NonEleSpace[#1][#2][#3][#4]-\CComp}

\newcommandx{\DLHier}[4][2=, 3=, 4=]
	{\mthargset[0]{\Delta#4}[#1][#3]{#2}\xspace}
\newcommandx{\DLHierH}[4][2=, 3=, 4=]
	{\DLHier{#1}[#2][#3][#4]-\HComp}
\newcommandx{\DLHierE}[4][2=, 3=, 4=]
	{\DLHier{#1}[#2][#3][#4]-\EComp}
\newcommandx{\DLHierC}[4][2=, 3=, 4=]
	{\DLHier{#1}[#2][#3][#4]-\CComp}

\newcommandx{\SLHier}[4][2=, 3=, 4=]
	{\mthargset[0]{\Sigma#4}[#1][#3]{#2}\xspace}
\newcommandx{\SLHierH}[4][2=, 3=, 4=]
	{\SLHier{#1}[#2][#3][#4]-\HComp}
\newcommandx{\SLHierE}[4][2=, 3=, 4=]
	{\SLHier{#1}[#2][#3][#4]-\EComp}
\newcommandx{\SLHierC}[4][2=, 3=, 4=]
	{\SLHier{#1}[#2][#3][#4]-\CComp}

\newcommandx{\PLHier}[4][2=, 3=, 4=]
	{\mthargset[0]{\Pi#4}[#1][#3]{#2}\xspace}
\newcommandx{\PLHierH}[4][2=, 3=, 4=]
	{\PLHier{#1}[#2][#3][#4]-\HComp}
\newcommandx{\PLHierE}[4][2=, 3=, 4=]
	{\PLHier{#1}[#2][#3][#4]-\EComp}
\newcommandx{\PLHierC}[4][2=, 3=, 4=]
	{\PLHier{#1}[#2][#3][#4]-\CComp}

\newcommandx{\DBHier}[4][2=, 3=, 4=]
	{\mthargset[3]{\Delta#4}[#1][#3]{#2}\xspace}
\newcommandx{\DBHierH}[4][2=, 3=, 4=]
	{\DBHier{#1}[#2][#3][#4]-\HComp}
\newcommandx{\DBHierE}[4][2=, 3=, 4=]
	{\DBHier{#1}[#2][#3][#4]-\EComp}
\newcommandx{\DBHierC}[4][2=, 3=, 4=]
	{\DBHier{#1}[#2][#3][#4]-\CComp}

\newcommandx{\SBHier}[4][2=, 3=, 4=]
	{\mthargset[3]{\Sigma#4}[#1][#3]{#2}\xspace}
\newcommandx{\SBHierH}[4][2=, 3=, 4=]
	{\SBHier{#1}[#2][#3][#4]-\HComp}
\newcommandx{\SBHierE}[4][2=, 3=, 4=]
	{\SBHier{#1}[#2][#3][#4]-\EComp}
\newcommandx{\SBHierC}[4][2=, 3=, 4=]
	{\SBHier{#1}[#2][#3][#4]-\CComp}

\newcommandx{\PBHier}[4][2=, 3=, 4=]
	{\mthargset[3]{\Pi#4}[#1][#3]{#2}\xspace}
\newcommandx{\PBHierH}[4][2=, 3=, 4=]
	{\PBHier{#1}[#2][#3][#4]-\HComp}
\newcommandx{\PBHierE}[4][2=, 3=, 4=]
	{\PBHier{#1}[#2][#3][#4]-\EComp}
\newcommandx{\PBHierC}[4][2=, 3=, 4=]
	{\PBHier{#1}[#2][#3][#4]-\CComp}

\newcommandx{\DPolHier}[4][2=, 3=, 4=]
	{\DLHier{#1}[#2][\argb{\mathrm{P}}{#3}][#4]}
\newcommandx{\DPolHierH}[4][2=, 3=, 4=]
	{\DPolHier{#1}[#2][#3][#4]-\HComp}
\newcommandx{\DPolHierE}[4][2=, 3=, 4=]
	{\DPolHier{#1}[#2][#3][#4]-\EComp}
\newcommandx{\DPolHierC}[4][2=, 3=, 4=]
	{\DPolHier{#1}[#2][#3][#4]-\CComp}

\newcommandx{\SPolHier}[4][2=, 3=, 4=]
	{\SLHier{#1}[#2][\argb{\mathrm{P}}{#3}][#4]}
\newcommandx{\SPolHierH}[4][2=, 3=, 4=]
	{\SPolHier{#1}[#2][#3][#4]-\HComp}
\newcommandx{\SPolHierE}[4][2=, 3=, 4=]
	{\SPolHier{#1}[#2][#3][#4]-\EComp}
\newcommandx{\SPolHierC}[4][2=, 3=, 4=]
	{\SPolHier{#1}[#2][#3][#4]-\CComp}

\newcommandx{\PPolHier}[4][2=, 3=, 4=]
	{\PLHier{#1}[#2][\argb{\mathrm{P}}{#3}][#4]}
\newcommandx{\PPolHierH}[4][2=, 3=, 4=]
	{\PPolHier{#1}[#2][#3][#4]-\HComp}
\newcommandx{\PPolHierE}[4][2=, 3=, 4=]
	{\PPolHier{#1}[#2][#3][#4]-\EComp}
\newcommandx{\PPolHierC}[4][2=, 3=, 4=]
	{\PPolHier{#1}[#2][#3][#4]-\CComp}

\newcommandx{\DAriHier}[4][2=, 3=, 4=]
	{\DLHier{#1}[#2][\argb{0}{#3}][#4]}
\newcommandx{\DAriHierH}[4][2=, 3=, 4=]
	{\DAriHier{#1}[#2][#3][#4]-\HComp}
\newcommandx{\DAriHierE}[4][2=, 3=, 4=]
	{\DAriHier{#1}[#2][#3][#4]-\EComp}
\newcommandx{\DAriHierC}[4][2=, 3=, 4=]
	{\DAriHier{#1}[#2][#3][#4]-\CComp}

\newcommandx{\SAriHier}[4][2=, 3=, 4=]
	{\SLHier{#1}[#2][\argb{0}{#3}][#4]}
\newcommandx{\SAriHierH}[4][2=, 3=, 4=]
	{\SAriHier{#1}[#2][#3][#4]-\HComp}
\newcommandx{\SAriHierE}[4][2=, 3=, 4=]
	{\SAriHier{#1}[#2][#3][#4]-\EComp}
\newcommandx{\SAriHierC}[4][2=, 3=, 4=]
	{\SAriHier{#1}[#2][#3][#4]-\CComp}

\newcommandx{\PAriHier}[4][2=, 3=, 4=]
	{\PLHier{#1}[#2][\argb{0}{#3}][#4]}
\newcommandx{\PAriHierH}[4][2=, 3=, 4=]
	{\PAriHier{#1}[#2][#3][#4]-\HComp}
\newcommandx{\PAriHierE}[4][2=, 3=, 4=]
	{\PAriHier{#1}[#2][#3][#4]-\EComp}
\newcommandx{\PAriHierC}[4][2=, 3=, 4=]
	{\PAriHier{#1}[#2][#3][#4]-\CComp}

\newcommandx{\DAnaHier}[4][2=, 3=, 4=]
	{\DLHier{#1}[#2][\argb{1}{#3}][#4]}
\newcommandx{\DAnaHierH}[4][2=, 3=, 4=]
	{\DAnaHier{#1}[#2][#3][#4]-\HComp}
\newcommandx{\DAnaHierE}[4][2=, 3=, 4=]
	{\DAnaHier{#1}[#2][#3][#4]-\EComp}
\newcommandx{\DAnaHierC}[4][2=, 3=, 4=]
	{\DAnaHier{#1}[#2][#3][#4]-\CComp}

\newcommandx{\SAnaHier}[4][2=, 3=, 4=]
	{\SLHier{#1}[#2][\argb{1}{#3}][#4]}
\newcommandx{\SAnaHierH}[4][2=, 3=, 4=]
	{\SAnaHier{#1}[#2][#3][#4]-\HComp}
\newcommandx{\SAnaHierE}[4][2=, 3=, 4=]
	{\SAnaHier{#1}[#2][#3][#4]-\EComp}
\newcommandx{\SAnaHierC}[4][2=, 3=, 4=]
	{\SAnaHier{#1}[#2][#3][#4]-\CComp}

\newcommandx{\PAnaHier}[4][2=, 3=, 4=]
	{\PLHier{#1}[#2][\argb{1}{#3}][#4]}
\newcommandx{\PAnaHierH}[4][2=, 3=, 4=]
	{\PAnaHier{#1}[#2][#3][#4]-\HComp}
\newcommandx{\PAnaHierE}[4][2=, 3=, 4=]
	{\PAnaHier{#1}[#2][#3][#4]-\EComp}
\newcommandx{\PAnaHierC}[4][2=, 3=, 4=]
	{\PAnaHier{#1}[#2][#3][#4]-\CComp}

\newcommandx{\DBorHier}[4][2=, 3=, 4=]
	{\DBHier{#1}[#2][\argb{\mathrm{B}}{#3}][#4]}
\newcommandx{\DBorHierH}[4][2=, 3=, 4=]
	{\DBorHier{#1}[#2][#3][#4]-\HComp}
\newcommandx{\DBorHierE}[4][2=, 3=, 4=]
	{\DBorHier{#1}[#2][#3][#4]-\EComp}
\newcommandx{\DBorHierC}[4][2=, 3=, 4=]
	{\DBorHier{#1}[#2][#3][#4]-\CComp}

\newcommandx{\SBorHier}[4][2=, 3=, 4=]
	{\SBHier{#1}[#2][\argb{\mathrm{B}}{#3}][#4]}
\newcommandx{\SBorHierH}[4][2=, 3=, 4=]
	{\SBorHier{#1}[#2][#3][#4]-\HComp}
\newcommandx{\SBorHierE}[4][2=, 3=, 4=]
	{\SBorHier{#1}[#2][#3][#4]-\EComp}
\newcommandx{\SBorHierC}[4][2=, 3=, 4=]
	{\SBorHier{#1}[#2][#3][#4]-\CComp}

\newcommandx{\PBorHier}[4][2=, 3=, 4=]
	{\PBHier{#1}[#2][\argb{\mathrm{B}}{#3}][#4]}
\newcommandx{\PBorHierH}[4][2=, 3=, 4=]
	{\PBorHier{#1}[#2][#3][#4]-\HComp}
\newcommandx{\PBorHierE}[4][2=, 3=, 4=]
	{\PBorHier{#1}[#2][#3][#4]-\EComp}
\newcommandx{\PBorHierC}[4][2=, 3=, 4=]
	{\PBorHier{#1}[#2][#3][#4]-\CComp}

\newcommand{\HComp}
	{\txtname{hard}\xspace}

\newcommand{\EComp}
	{\txtname{easy}\xspace}

\newcommand{\CComp}
	{\txtname{complete}\xspace}

\newcommand{\WinSet}{\mthset{Win}}

\newcommand{\priFun}{\mthfun{p}}

\newcommand{\force}{\mthfun{force}}
\newcommand{\fforce}{\force_{0}}
\newcommand{\sforce}{\force_{1}}

\newcommand{\Inf}{\mthfun{Inf}}
\renewcommand{\bar}[1]{\overline{#1}}

\newcommand{\PGSolver}{\mthsym{PGSolver}}

\newcommand{\RE}{\mthsym{RE}}
\newcommand{\SP}{\mthsym{SP}}
\renewcommand{\APT}{\mthsym{APT}}

\newcommand{\abortT}{\mthsym{abort}[T]}
\newcommand{\abortM}{\mthsym{abort}[M]}
\usepackage{tikz}
\usetikzlibrary{arrows,shapes}

\usepackage{wrapfig}
\usepackage[caption=false,font=footnotesize]{subfig}

\tikzstyle{every node} =
	[draw = none, fill = white, thin]
\tikzstyle{every edge} +=
	[black, thick]

\tikzstyle{noall} =
	[draw = none, fill = none]
\tikzstyle{nodraw} =
	[draw = none, fill = white]
\tikzstyle{nofill} =
	[draw = black, fill = none]

\tikzstyle{cnode} =
	[circle, draw = black]
\tikzstyle{snode} =
	[regular polygon, regular polygon sides = 4, draw = black]
\tikzstyle{lnode} =
	[diamond, draw = gray!75]
\tikzstyle{pnode} =
	[regular polygon, regular polygon sides = 5, draw = gray]

\AfterEndPreamble
	{

	\newcommand{\figexmgam}
		{
		\begin{wrapfigure}[11]{r}{0.4\textwidth}
			\vspace{-2.1em}
			\hspace{-2.1em}
				\footnotesize
				\scalebox{0.7}[0.7]
					{
					\begin{tikzpicture}[node distance = 5em, bend angle = 30]

					\node [cnode]
							(5)
							[]
							{$\stackrel{5}{\qSym[5]}$};
					\node [cnode]
							(0)
							[left of = 5]
							{$\stackrel{3}{\qSym[0]}$};
					\node [cnode]
							(3)
							[right of = 5]
							{$\stackrel{2}{\qSym[3]}$};
					\node [snode]
							(1)
							[right of = 3]
							{$\stackrel{1}{\qSym[1]}$};
					\node [snode]
							(2)
							[below of = 5]
							{$\stackrel{5}{\qSym[2]}$};
					\node [cnode]
							(4)
							[below of = 2]
							{$\stackrel{2}{\qSym[4]}$};
					\node [snode]
							(6)
							[right of = 4]
							{$\stackrel{2}{\qSym[6]}$};
					\path[->]
							(0)		edge	[]
													(5)
										edge	[]
													(2)
							(1)		edge	[bend angle = 45, bend right]
													(5)
										edge	[]
													(2)
										edge	[]
													(3)
							(2)		edge	[]
													(4)
										edge	[]
													(6)
										edge	[]
													(3)
							(3)		edge	[bend right]
													(5)
										edge	[]
													(6)
										edge 	[loop left]
													()
							(4)		edge	[bend right]
													(6)
							(5)		edge	[]
													(2)
							(6)		edge	[bend right]
													(4)
										edge	[loop right]
													()
										edge	[bend angle = 110, bend left]
													(0)
							;

					\end{tikzpicture}
					}
				\vspace{-4.2em}
				\caption{\label{fig:pargame} \small A parity game.}
			\vspace{-8em}
		\end{wrapfigure}
		}
	\noindent

	}

\usepackage{multirow}

% ...

\AfterEndPreamble
	{

	}

\usepackage[ruled,vlined]{algorithm2e}

\SetKw{Signature}{signature}
\SetKwFor{Function}{function}{}{}

\AfterEndPreamble
	{

	}

\begin{document}
	\pagenumbering{arabic}
 \pagestyle{plain}
	\title{Solving Parity Games Using An Automata-Based Algorithm \thanks{An earlier version of the paper appeared in \cite{DMPV16}.			This version corrects a minor error in that earlier version. See Page 6.} \thanks{Work supported by NSF grants CCF-1319459 and IIS-1527668, NSF
			Expeditions in Computing project "ExCAPE: Expeditions in Computer
			Augmented Program Engineering", BSF grant 9800096,  ERC Advanced Investigator
			Grant 291528 (``Race'') at Oxford and 		
			GNCS 2016: Logica, Automi e Giochi per Sistemi Auto-adattivi.} 
}

	\author
		{%
		Antonio Di Stasio$^{1}$, Aniello Murano$^{1}$, Giuseppe
		Perelli$^{2}$~\thanks{Part of the work has been done while visiting Rice University.}, Moshe Y. Vardi$^{3}$
		}

	\institute{$^{1}$Universit\`a di Napoli ``Federico II'', $^{2}$University of
	Oxford $^{3}$Rice University}

	\maketitle

	% Begin of file Abstract.tex

\begin{abstract}

	\emph{Parity games} are abstract infinite-round games that take an
	important role in formal verification.
% 	Notably, several questions related to verification and synthesis of reactive
% 	systems, alternating $\omega$-automata, and $\mu$-calculus can be polynomially
% 	rephrased in terms of these games and \emph{vice versa}.
 	In the basic setting, these games are two-player, turn-based, and played under
	perfect information on directed	graphs, whose nodes are labeled with
	priorities.
	The winner of a play is	determined according to the parities (even or odd) of
	the minimal priority occurring infinitely often in that play.
	The problem of finding a winning strategy in parity games is known to be in
	\UPTime $\cap$ \CoUPTime and deciding whether a polynomial time solution
	exists is a long-standing open question.
	In the last two decades, a variety of algorithms
%	and successive optimizations
	have been proposed. Many of them have been also implemented in a platform
	named \PGSolver.
	%, which has made possible to evaluate the performance in several specific
	%case, rather than the asymptotically.
        %MYV:
        This has enabled an empirical evaluation of these algorithms and a better
        understanding of their relative merits.
	%This has led to bring to the fore algorithms under precise parameter
	%settings.

	In this paper, we further contribute to this subject by implementing, for
	the first time, 
        %MYV:
        %in \PGSolver, 
        an algorithm based on alternating automata.
	More precisely, we consider an algorithm introduced by	Kupferman and Vardi
	that solves a parity game by solving the emptiness problem of a corresponding
	alternating parity automaton.
        %MYV:
        Our empirical evaluation demonstrates that
	%The main result we gain from the benchmarks is that 
        this algorithm outperforms other algorithms when the game has a 
        %MYV:
        a small number of priorities relative to the size of the game.
        %number of priorities up to a logarithmic in the number of nodes. 
        In many concrete applications, we do indeed end up with parity games
	where the number of priorities is relatively small. This makes
	the new algorithm quite useful in practice.
\vspace{-1em}
\end{abstract}

% End of file Abstract.tex

	% Begin of file Introduction.tex

\begin{section}{Introduction}
	\emph{Parity games}~\cite{EJ91,Zie98}
	are abstract infinite-duration games that represent a
	powerful mathematical framework to address fundamental questions in computer
	science.  They are intimately related to other infinite-round games, such
	as \emph{mean} and \emph{discounted} payoff, \emph{stochastic}, and
	\emph{multi-agent} games~\cite{CHJ05,CDHR10,CJH04,Ber07}.
	\\ \indent
	In the basic setting, parity games are two-player, turn-based, played on
	directed graphs whose nodes are labeled with priorities (also called,
	\emph{colors}) and players have perfect information about the adversary moves.
	The two players, Player~$\frsentnam$ and Player~$\scnentnam$, take turns moving
	a token along the edges of the graph starting from a designated initial node.
	Thus, a play induces an infinite path and Player~$\frsentnam$ wins the play if
	the smallest priority visited infinitely often is even;
  otherwise, Player~$\scnentnam$ wins the play.
	The problem of deciding if Player~$\scnentnam$ has a winning strategy (i.e.,
	can induce a winning play) in a given parity game is known to be in \UPTime
	$\cap$ \CoUPTime~\cite{Jur98};
	whether a polynomial time solution exists is a long-standing open
	question~\cite{Wilke01}.
	
	Several algorithms for solving parity games have been proposed in the last two
	decades, aiming to tighten the known complexity bounds for the problem, as
	well as come out with solutions that work well in practice.  Among the latter,
	we recall the recursive algorithm (\RE) proposed by Zielonka~\cite{Zie98},
	%, the local $\mu$-calculus model checker by Stevens and Stirling [13],
	the Jurdzi\'nski's small-progress measures algorithm~\cite{Jur00} (\SP),
% 	along with its symbolic version via a reduction to SAT~\cite{HKLN12},
	the strategy-improvement algorithm by Jurdzi\'nski and V\"oge~\cite{VJ00},
% 	along with its local optimal variation by Schewe~\cite{Sch08}, the dominion decomposition
	the (subexponential) algorithm by Jurdzi\'nki, Paterson and
	Zwick~\cite{JPZ08}, and the big-step algorithm by Schewe~\cite{Sch07}.
	These algorithms have been implemented in the platform \PGSolver, and extensively
	investigated experimentally~\cite{FL09,FL09PGSolver}.
	This study has also led to a few key optimizations, such as the
	decomposition into strongly connected components, the removal of
	self-cycles on nodes, and the application of a priority
	compression~\cite{ACH09,Jur00}.
	Specifically, the latter allows to reduce a game to an equivalent game where the
	priorities are replaced in such a way they form a dense sequence of natural
	numbers, $1,2, \ldots, d$, for a minimal possible $d$.
	Table~\ref{tab:parityalg} summarizes the mentioned algorithms along with
	their known worst-case complexity, where the parameters $n$, $e$, and
	$d$ denote the number of nodes, edges, and priorities, respectively
	(see~\cite{FL09,FL09PGSolver}, for more).
	
	\begin{table}[htb]
		\vspace{-1em}
			\begin{center}
				\scalebox{0.9}[0.9]
					{
					\begin{tabular}{|l||l|}
						\hline
%------------------------------------------------------------------%
						Algorithm &
						Computational Complexity\\
%------------------------------------------------------------------%
						\hline
						\hline
%------------------------------------------------------------------%
		  Recursive (RE)~\cite{Zie98} & $O(e \cdot n^d)$\\
		  Small Progress Measures (SP)~\cite{Jur00} & $O(d \cdot e \cdot
		  (\frac{n}{d})^{\frac{d}{2}})$\\
% 		  Small Progress Measures via SAT(SPMS)~\cite{HKLN12} & $O(e \cdot d)$ +
% 			O(SAT)\\
		  Strategy Improvement (SI)~\cite{VJ00} & $O(2^e \cdot n \cdot e)$\\
% 		  Optimal Strategy Improvement (OSI)~\cite{Sch08}& $O(e \cdot
% 			(\frac{n+d}{d})^d \cdot log(\frac{n+d}{d}))$\\
		  Dominion Decomposition (DD)~\cite{JPZ08}& $O(n^{\sqrt{n}})$\\
		  Big Step (BS)~\cite{Sch07}& $O(e \cdot n^{\frac{1}{3}d})$ \\

%------------------------------------------------------------------%
						\hline
					\end{tabular}
					}
					\vspace{0.5em}
					\caption{\label{tab:parityalg} \small{Parity algorithms along with
						their computational complexities.}}
			\end{center}
			\vspace{-2.5em}
		\end{table}
	In formal system design ~\cite{CGP02,CE81,KVW00,QS81}, parity
	games arise as a natural evaluation machinery for the automatic synthesis and
	verification of distributed and reactive
	systems~\cite{KVW01,Thomas09,AKM12}, as they allow to express liveness
	and safety properties in a very elegant and powerful way\cite{MMS13}.
	Specifically, in model-checking, one can check the correctness of a
	system with respect to a desired behavior, by checking whether a model of the
	system, that is, a \emph{Kripke structure}, is correct with respect to a
	formal specification of its behavior, usually described in terms of a modal
	logic formula.  In case the specification is given as a $\mu$-calculus
	formula~\cite{Koz83}, the model checking question can be rephrased, in
	linear-time, as a parity game~\cite{EJ91}.  So, a parity game solver
	can be used as a model checker for a $\mu$-calculus specification (and
	vice-versa), as well as for fragments such as \CTL,  \CTLS, and the like.
	\\ \indent
	In the automata-theoretic approach to $\mu$-calculus model checking, under a
	linear-time translation, one can also reduce the verification problem to a
	question about automata.
	More precisely, one can take the product of the model and an alternating tree
	automaton accepting all tree models of the specification. This product can be
	defined as an alternating word parity automaton over a singleton alphabet, and
	the system is correct with respect to the specification iff this automaton is
	nonempty \cite{KVW00}.
	It has been proved there that the nonemptiness problems for nondeterministic
	tree parity automata and alternating word parity automata over a singleton
	alphabet are equivalent and that their complexities coincide.
	For this reason, in the sequel we refer to these two kinds of automata just as
	parity automata.
	Hence, algorithms for the solution of the $\mu$-calculus model checking
	problem, parity games, and the emptiness problem for parity automata can be
	interchangeably used to solve any of these problems, as they are linear-time
	equivalent.
	Several algorithms have been proposed in the literature to solve the
	non-emptiness problem of parity automata, but none of them has been ever
 	implemented under the purpose of solving parity games.
% 	Conversely, there are few practical interactions
%	between	model checking and parity games\cite{PC98}.
%
%
% 	We recall that in an alternating automata [BL80, CKS81], both existential
% 	and	universal branching modes are allowed, and transitions are given as
% 	positive Boolean formulas over the set of states.
% 	The rich combinatorial structure of alternating automata makes the
% 	translation of
% 	specifications much simpler	than translating them to nondeterministic
% 	automata
% 	[Var94].
	%
	\\ \indent
	In this paper, we study and implement an algorithm, which we call \APT,
	introduced by Kupferman and Vardi in~\cite{KV98}, for solving parity games via
	emptiness checking of alternating parity automata, and evaluate its performance 
        over the \PGSolver\ platform.
	This algorithm has been sketched in~\cite{KV98}, but not spelled out in detail
	and without a correctness proof, two major gaps that we fill here.
	The core idea of the \APT\ algorithm is an efficient translation to \emph{weak
	alternating automata}~\cite{MSS88}.
	These are a special case of B\"uchi automata in which the set of states is
	partitioned into partially ordered sets.
	Each set is classified as accepting or rejecting.
	The transition function is restricted so that the automaton either stays at
	the same set or moves to a smaller set in the partial order.
	Thus, each run of a weak automaton eventually gets trapped in some set in the
	partition.
%	Acceptance is then determined according to the classification of this set.
%
	The special structure of weak automata is reflected in their attractive
	computational properties.
	In particular, the nonemptiness problem for weak automata can be solved in
	linear time~\cite{KVW00}, while the best known upper bound for the
	nonemptiness problem for B\"uchi automata is quadratic \cite{ChatHenz12}.
	Given an alternating parity word automaton with $n$ states and $d$ colors, the
	\APT\ algorithm checks the emptiness of an equivalent weak alternating word
	automaton with O($n^d$) states.
	The construction goes through a sequence of $d$ intermediate automata.
	Each automaton in the sequence refines the state space of its predecessor and
	has one less color to check in its parity
				%acceptance
	condition.
	Since one can check in linear time the emptiness of such an automaton, we get
	an $O(n^d)$ overall complexity for the addressed problem. \APT\ does not
	construct the equivalent weak automaton directly, but applies the emptiness
	test directly, constructing the equivalent weak automaton on the fly.
	\\ \indent
	We evaluated our implementation of the \APT\ algorithm over several random
	game instances,
        %The experiments let us to observe interesting aspects regarding the
	%performance under specific parameters as well as correct some little flaws in
	%its automata instantiation.
	%Regarding the latter, the most important observation regards the fact that the
	%version of the \emph{APT} algorithm given in \cite{KV98} can only work if the
	%priorities are compact.
	%Note that this concept is analogous to the one used in parity games.
	%However, it is worth observing that, for the latter, it represents an
	%optimization, while for the \emph{APT} algorithm, it is a correction.
	%As far as the experiments regard, we compared the \emph{APT} algorithm with
	comparing it with \RE\ and \SP\ algorithms.
	%By the benchmarks we observe that in the case
	Our main finding is that when the number of the priority in a game is
	significantly smaller (specifically, logarithmically) than the number of nodes in
	the game graph, 
        %the arena has a number of priorities up to a (natural) logarithmic in the number of nodes, 
	the \APT\ algorithm significantly outperform the other algorithms.
	We take this as an important development since in many real applications of
	parity games we do get game instances where the number of priorities is indeed
	very small compared to the size of the game graph.  For example, coming back
	to the automata-theoretic approach to $\mu$-calculus model
	checking~\cite{KVW00}, the translation usually results in a parity automaton
	(and thus in a parity game) with few priorities, but with a huge number of
	nodes.
	This is due to the fact that usually specification formulas are small, while
	the system is big.
	A similar phenomenon occurs in the application of parity games to reactive
	synthesis \cite{Thomas09}.
	\\ \indent
	\textbf{Outline}
	The sequel of the paper is as follows.
	Section~2 gives preliminary concepts on parity games.
	Section~3 introduces extended parity games and describes the \APT\ algorithm
	in detail, including a proof of correctness.
	Section~4 describes the implementation of the \APT\ algorithm in the tool \PGSolver.
	Section~5 contains the experimental results on runtime for \APT\ over random
	benchmarks.
	Finally, Section~6 gives some conclusions.
	%
	%\\ \indent
	%
%
% Due to the lack of space, proofs are omitted and reported in the full
%

\end{section}

% End of file Introduction.tex

% Begin of file Preliminaries.tex

\begin{section}{Preliminaries}
	\label{sec:prl}

	In this section, we briefly recall some basic concepts regarding parity games.
	A \emph{Parity Game} (\PG, for short) is a tuple $\GmName \defeq \ArStr$,
	where $\FPsSet$ and $\SPsSet$ are two finite disjoint sets of nodes for
	Player~$\frsentnam$ and Player~$\scnentnam$, respectively, with $\PsSet =
	\FPsSet \cup \SPsSet$, $\MvRel \subseteq \PsSet \times \PsSet$, is the
	left-total binary relation of moves, and $\priFun: \PsSet \to \SetN$ is the
	priority function~\footnote{Here, we mean the set of non-negative integers,
	excluding zero.}.
	Each player moves a token along nodes by means of the relation $\MvRel$.
	By $\MvRel(\psElm) \defeq \set{\psElm' \in \PsSet}{(\psElm, \psElm') \in
	\MvRel}$ we denote the set of nodes to which the token can be moved, starting
	from node $q$.
% 	and $\EdgSet(\psElm) \defeq \set{(\psElm', \psElm'') \in \MvRel}{
% 	\psElm' = \psElm}$
% 		we denote the set of $\MvRel$-successors.
% 	and outgoing
% 	$\MvRel$-edges of $\psElm$, respectively.

%\setlength{\intextsep}{0pt}%
\setlength{\columnsep}{-2pt}%

	\figexmgam

	As a running example, consider the \PG depicted in
	Figure~\ref{fig:pargame}. The set of players's nodes is
	$\FPsSet = \{\qSym[0], \qSym[3],	\qSym[4],	\qSym[5]\}$ and $\SPsSet =
	\{\qSym[1], \qSym[2], \qSym[6] \}$; we use circles to denote
	nodes belonging to Player~$\frsentnam$ and squares	for
	those belonging to Player~$\scnentnam$.
	$\MvRel$ is described by arrows. Finally,	the priority function $\priFun$ is
	given by $\priFun(\qSym[1]) = 1$,	$\priFun(\qSym[3]) = \priFun(\qSym[4]) =
	\priFun(\qSym[6]) = 2$, $\priFun(\qSym[0]) = 3$, and $\priFun(\qSym[2]) =
	\priFun(\qSym[5]) = 5$.

	A \emph{play} (\resp \emph{history}) over $\GmName$ is an infinite (\resp
	finite) sequence $\playElm = \psElm[1] \cdot \psElm[2] \cdot \ldots \in
	\PthSet \subseteq \PsSet[][\omega]$ (\resp $\hstElm = \psElm[1] \cdot \ldots
	\cdot \psElm[n] \in \HstSet \subseteq \PsSet[][*]$) of nodes that agree
	with $\MvRel$, \ie, $(\playElm[i], \playElm[i + 1]) \in \MvRel$, for each
	natural number $i \in \SetN$ (\resp $i \in \numcc{1}{n - 1}$).
	In the \PG in Figure~\ref{fig:pargame}, a possible play is $\bar{\playElm} =
	\qSym[1] \cdot \qSym[5] \cdot \qSym[2] \cdot (\qSym[3])^{\omega}$, while a
	possible history is given by $\bar{\hstElm} = \qSym[1] \cdot \qSym[5] \cdot
	\qSym[2] \cdot \qSym[3]$.
	
	For a given play $\playElm = \psElm[1] \cdot \psElm[2] \cdot \ldots$, by
	$\priFun(\playElm) = \priFun(\psElm[1]) \cdot \priFun(\psElm[2]) \cdot \ldots
	\in \SetN[][\omega]$ we denote the associated priority sequence.
	As an example, the associated priority sequence to $\bar{\playElm}$ is given
	by $\priFun(\bar{\playElm}) = 1 \cdot 5 \cdot 5 \cdot (2)^{\omega}$.
	
	For a given history $\hstElm = \psElm[1] \cdot \ldots \cdot \psElm[n]$, by
	$\fst{\hstElm}\defeq \psElm[1]$ and $\lst{\hstElm} \defeq \psElm[n]$ we denote
	the first and last node occurring in $\hstElm$, respectively.
	For the example history, we have that $\fst{\bar{\hstElm}} = \qSym[1]$ and
	$\lst{\bar{\hstElm}} = \qSym[3]$.
	By $\FHstSet$ (\resp $\SHstSet$) we denote the set of histories $\hstElm$ such
	that $\lst{\hstElm} \in \FPsSet$ (\resp $\lst{\hstElm} \in \SPsSet$).
	Moreover, by $\Inf(\playElm)$ and $\Inf(\priFun(\playElm))$ we denote the set
	of nodes and priorities that occur infinitely often in $\playElm$ and
	$\priFun(\playElm)$, respectively.
	Finally, a play $\playElm$ is winning for Player $\frsentnam$  (resp., Player
	$\scnentnam$) if $\min(\Inf(\priFun(\playElm)))$ is even (\resp odd).
	In the running example, we have that $\Inf(\bar{\playElm}) = \{\qSym[3]\}$ and
	$\Inf(\priFun(\bar{\playElm})) = \{2\}$ and so, $\playElm$ is winning for
	Player $\frsentnam$.
	
	A Player $\frsentnam$ (\resp Player $\scnentnam$) strategy is a function
	$\fstrElm: \FHstSet \to \PsSet$ (\resp $\sstrElm: \SHstSet \to \PsSet$) such
	that, for all $\hstElm \in \FHstSet$ (\resp $\hstElm \in \SHstSet$), it holds
	that $(\lst{\hstElm}, \fstrElm(\hstElm)) \in \MvRel$ (\resp $\lst{\hstElm},
	\sstrElm(\hstElm)) \in \MvRel$).

	Given a node $\psElm$, Player $\frsentnam$ and a Player $\scnentnam$
	strategies $\fstrElm$ and $\sstrElm$, the play of these two strategies,
	denoted by $\playFun(\psElm, \fstrElm, \sstrElm)$, is the only play $\playElm$
	in the game that starts in $\psElm$ and agrees with both Player $\frsentnam$
	and Player $\scnentnam$ strategies, \ie, for all $i \in \SetN$, if
	$\playElm[i] \in \FPsSet$, then $\playElm[i + 1] = \fstrElm(\playElm[i])$, and
	$\playElm[i + 1] = \sstrElm(\playElm[i])$, otherwise.
	
	A strategy $\fstrElm$ (\resp $\sstrElm$) is \emph{memoryless} if, for all
	$\hstElm[1], \hstElm[2] \in \FHstSet$ (\resp $\hstElm[1], \hstElm[2] \in
	\SHstSet$), with $\lst{\hstElm[1]} = \lst{\hstElm[2]}$, it holds that
	$\fstrElm(\hstElm[1]) = \fstrElm(\hstElm[2])$ (\resp $\fstrElm(\hstElm[1]) =
	\sstrElm(\hstElm[2])$).
	Note that a memoryless strategy can be defined on the set of nodes, instead of
	the set of histories.
	Thus we have that they are of the form $\fstrElm: \FPsSet \to \PsSet$ and
	$\sstrElm: \SPsSet \to \PsSet$.

	We say that Player $\frsentnam$ (\resp Player $\scnentnam$) \emph{wins} the
	game $\GmName$ from node $\psElm$ if there exists a Player $\frsentnam$
	(\resp Player $\scnentnam$) strategy $\fstrElm$ (\resp $\sstrElm$) such that,
	for all Player $\scnentnam$ (\resp Player $\frsentnam$) strategies $\sstrElm$
	(\resp $\fstrElm$) it holds that $\playFun(\psElm, \fstrElm, \sstrElm)$ is
	winning for Player $\frsentnam$ (\resp Player $\scnentnam$).

	A node $\psElm$ is \emph{winning} for Player $\frsentnam$ (\resp Player
	$\scnentnam$) if Player $\frsentnam$ (\resp Player $\scnentnam$) wins the game
	from $\psElm$.
	By $\WinSet[0](\GmName)$ (\resp $\WinSet[1](\GmName)$) we denote the set of
	winning nodes in $\GmName$ for Player $\frsentnam$ (\resp Player
	$\scnentnam$).
	Parity games enjoy determinacy, meaning that, for every node $q$, either $q
	\in \WinSet[0](\GmName)$ or $q \in \WinSet[1](\GmName)$~\cite{EJ91}.
	Moreover, it can be proved that, if Player $\frsentnam$ (resp., Player
	$\scnentnam$) has a winning strategy from node $q$, then it has a memoryless
	winning strategy from the same node~\cite{Zie98}.

\end{section}

% End of file Preliminaries.tex

	% Begin of file SectionI.tex

\begin{section}{Extended Parity Games}
	\label{sec:pargam}
	%introduzione all'algoritmo
	In this section we recall the \APT\ algorithm,
	introduced	by Kupferman and Vardi in~\cite{KV98}, to solve parity games	via
	emptiness checking of parity automata. More important, we fill two major	gaps
	from~\cite{KV98} which is to spell out in	details the
	definition of the \APT\ algorithm as well as to give a	correctness proof.
	The \APT\ algorithm makes use of two special (incomparable) sets of nodes,
	denoted by $\VSet$ and $\ASet$, and called set	of \emph{Visiting} and
	\emph{Avoiding}, respectively.
	Intuitively, a node is declared visiting for a player at the stage in
	which it is clear	that, by reaching that node, he can surely induce a
	winning play and thus	winning the game.
	Conversely, a node is declared avoiding for a player whenever it is clear
	that, by reaching that node, he is not able to induce any winning play and
	thus losing the game. The	algorithm, in turns, tries to
	partition all nodes of the game into these two	sets.
	The formal definition of the sets $V$ and $A$ follows.
	
	An \emph{Extended Parity Game}, (\EPG, for short) is a tuple $\ArStr[EPG]$
	where $\FPsSet$, $\SPsSet$, $\MvRel$ are as in \PG. The subsets of
	nodes $\VSet,	\ASet \subseteq 	\PsSet = \FPsSet \cup \SPsSet$ are two
	disjoint sets of	\emph{Visiting} and \emph{Avoiding} nodes, respectively.
	Finally, $\pFun: \PsSet \to \SetN$ is a parity function mapping every
	non-visiting and non-avoiding	set to a color.
	
	The notions of histories and plays are equivalent to the ones
	given for \PG. Moreover, as far as the definition of strategies is concerned,
	we say that a play $\playElm$ that is in $\PsSet \cdot (\PsSet \setminus
	(\VSet \cup \ASet))^{*} \cdot \VSet  \cdot \PsSet[][\omega]$ is winning for
	Player $\frsentnam$, while a play $\playElm$ that is in $\PsSet \cdot (\PsSet
	\setminus (\VSet \cup \ASet))^{*} \cdot \ASet \cdot \PsSet[][\omega]$ is
	winning for Player $\scnentnam$.
	For a play $\playElm$ that never hits either $\VSet$ or $\ASet$, we say that
	it is winning for Player $\frsentnam$ iff it satisfies the parity condition,
	\ie, $\min(\Inf(\pFun(\playElm)))$ is even, otherwise it is winning for 
	Player $\scnentnam$.
	
	Clearly, \PG{s} are special cases of \EPG{s} in which $\VSet = \ASet =
	\emptyset$.
	Conversely, one can transform an \EPG into an equivalent \PG with the
	same winning set by simply replacing every outgoing edge with loop to every
	node in $\VSet \cup \ASet$ and then relabeling each node in $\VSet$ and
	$\ASet$ with an even and an odd number, respectively.
	
	In order to describe how to solve \EPG{s}, we introduce some notation.
	By $\FSet[i] = \pFun[][-1](i)$ we denote the set of all nodes labeled with
	$i$.
	Doing that, the parity condition can be described as a finite sequence $\alpha
	= \FSet[1] \cdot \ldots \cdot \FSet[k]$ of sets, which alternates from sets of
	nodes with even priorities to sets of nodes with odd priorities and the other
	way round, forming a partition of the set of nodes, ordered by the priority
	assigned by the parity function.
	We call the set of nodes $F_i$ an even (\resp odd) parity set if $i$ is even
	(\resp odd).
	
	For a given set $\XSet \subseteq \PsSet$, by $\fforce(\XSet) = \set{\qElm
	\in \FPsSet}{\XSet \cap \MvRel(\qElm) \neq \emptyset} \cup \set{\qElm \in
	\SPsSet}{\XSet \subseteq \MvRel(\qElm)}$ we denote the set of nodes from
	which Player $\frsentnam$ can force to move in the set $\XSet$.
	Analogously, by $\sforce(\XSet) = \set{\qElm
	\in \SPsSet}{\XSet \cap \MvRel(\qElm) \neq \emptyset} \cup \set{\qElm \in
	\FPsSet}{\XSet \subseteq \MvRel(\qElm)}$ we denote the set of nodes from
	which Player $\scnentnam$ can force to move in the set $\XSet$. 
	For example, in the \PG in Figure~\ref{fig:pargame}, $\sforce(\{ \qSym[6]\}) =
	\{\qSym[2], \qSym[4],\qSym[6]\}$.

% 	...
% 	
% 	\begin{lemma}
% 		\label{lmm:winchar}
% 		Let $\GName = \ArStr[EPG]$ be an \EPG, and $\qElm \in \PsSet$ a position.
% 		Then, the following holds.
% 		
% 		\begin{enumerate}
% 			\item
% 				If $\qElm \in \WinSet[0](\GName)$, then $\WinSet[0](\GName) =
% 				\WinSet[0](\GName')$, where $\GName' =
% 				\tuplef{\FPsSet}{\SPsSet}{\VSet \cup \{\qElm\}}{\ASet}{\MvRel}{\pFun}$;
% 				
% 			\item
% 				If $\qElm \in \WinSet[1](\GName)$, then $\WinSet[1](\GName) =
% 				\WinSet[1](\GName')$, where $\GName' =
% 				\tuplef{\FPsSet}{\SPsSet}{\VSet}{\ASet \cup \{\qElm\}}{\MvRel}{\pFun}$.
% 		\end{enumerate}
% 
% 	\end{lemma}

	We now introduce two functions that are co-inductively defined that will be
	used to compute the winning sets of Player $\frsentnam$ and Player
	$\scnentnam$, respectively.
	
	For a given \EPG $\GName$ with $\alpha$ being the representation of its parity
	condition, $\VSet$ its visiting set, and $\ASet$ its avoiding set, we define
	the functions $\WinSet[0](\alpha, \VSet, \ASet)$ and $\WinSet[1](\alpha,
	\ASet, \VSet)$.
	Informally, $\WinSet[0](\alpha, \VSet, \ASet)$ computes the set of nodes
	from which the player $\frsentnam$ has a strategy that avoids $\ASet$ and
	either force a visit in $\VSet$ or he wins the parity condition.
	The definition is symmetric for the function $\WinSet[1](\alpha, \ASet,
	\VSet)$. 
	Formally, we define $\WinSet[0](\alpha, \VSet, \ASet)$ and $\WinSet[1](\alpha,
	\ASet, \VSet)$ as follows.
	
	If $\alpha = \varepsilon$ is the empty sequence, then
	
	\begin{itemize}
		\item
			$\WinSet[0](\varepsilon, \VSet, \ASet) = \fforce(\VSet)$ and
	 
		\item
			$\WinSet[1](\varepsilon, \ASet, \VSet) = \sforce(\ASet)$.
	\end{itemize}

	Otherwise, if $\alpha = \FSet \cdot \alpha'$, for some set $\FSet$, then
	
	\begin{itemize}
		\item
			$\WinSet[0](\FSet \cdot \alpha', \VSet, \ASet) = \mu \YSet (\PsSet \setminus
			\WinSet[1](\alpha', \ASet \cup (\FSet \setminus \YSet), \VSet
			\cup (\FSet \cap \YSet)))$~\footnote{In~\cite{DMPV16} the set $\WinSet[0]$ is mistakenly typed as $\PsSet \setminus \mu\YSet(\WinSet[1](\alpha', \ASet \cup (\FSet \setminus \YSet), \VSet \cup (\FSet \cap \YSet)))$.
%			This is wrong since $\YSet$ is the winning set for Player 0, instead it is considered winning for Player 1 and it wrongly moves nodes from visiting to avoiding for Player 1.
			Here, we provide its correct formulation.
		 	Please, note that in the proof of~\cite[Theorem 1]{DMPV16}, the formula is correctly reported.} and
	 
		\item
			$\WinSet[1](\FSet \cdot \alpha', \ASet, \VSet) =  \mu \YSet (\PsSet \setminus
			\WinSet[0](\alpha', \VSet \cup (\FSet \setminus \YSet), \ASet
			\cup (\FSet \cap \YSet)))$,

	\end{itemize}
	where $\mu$ is the least fixed-point operator\footnote{The unravelling of
	$\WinSet[0]$ and $\WinSet[1]$ has some analogies with the fixed-point
	formula introduced in~\cite{Wal96} also used to solve parity games.
	Unlike our work, however, the formula presented there is just a translation of
	the Zielonka's algorithm~\cite{Zie98}.}.

	To better understand how \APT\ solves a parity game we show a simple piece of
	execution on the example in Fig \ref{fig:pargame}.
	It is easy to see that such parity game is won by Player $\frsentnam$ in all
	the possible starting nodes.
	Then, the fixpoint returns the entire set $\PsSet$.
	The parity condition is given by  $\alpha = \FSet_1 \cdot \FSet_2 \cdot
	\FSet_3 \cdot \FSet_4 \cdot \FSet_5$, where $\FSet_1 = \{\qSym[1]\}$,
	$\FSet_2 = \{ \qSym[3], \qSym[4], \qSym[6]\}$, $\FSet_3 = \{ \qSym[0] \}$, 
	$\FSet_4 = \emptyset$, $\FSet_5 = \{ \qSym[5], \qSym[6] \}$.
	The repeated application of functions $\WinSet[0](\alpha, \VSet,
	\ASet)$ and $\WinSet[1](\alpha, \ASet, \VSet)$ returns:

	\vspace{-0.5em}	
	$$\WinSet[0](\alpha, \emptyset, \emptyset) = \mu \YSet^1 ( \PsSet \setminus \mu
	\YSet^2(\PsSet \setminus \mu \YSet^3 (\PsSet \setminus \mu \YSet^4 ( \PsSet
	\setminus \mu \YSet^5(\PsSet \setminus 	\sforce(\VSet^6))))))$$

	in which the sets $\YSet[][i]$ are the nested fixpoint of the formula, while
	the set $\VSet[][6]$ is obtained by recursively applying the following:

	\begin{itemize}
		\item
			$\VSet^1 = \emptyset$, $\VSet^{i + 1} = \ASet^i \cup (\FSet_i \setminus
			\YSet^i)$, and
		\item
			$\ASet^1 = \emptyset$, $\ASet^{i + 1} = \VSet^i \cup (\FSet_i \cap
			\YSet^i)$.
	\end{itemize}
	
	As a first step of the fixpoint computation, we have that
	$\YSet^1 = \YSet^2 = \YSet^3 = \YSet^4 = \YSet^5 = \emptyset$.
	Then, by following the two iterations above for the example in
	Figure~\ref{fig:pargame}, we obtain that $\VSet[][6] = \{ \qSym[0], \qSym[1],
	\qSym[2], \qSym[5] \}$.
%        \begin{itemize}
%         \item $\VSet^2=\ASet^1 \cup (\FSet_1 \setminus \YSet^1)=\{\qSym[1]\}$,
%         \item $\ASet^2=\VSet^1 \cup (\FSet_1 \cap \YSet^1) = \emptyset$.
%        \end{itemize}       
%        At the third step:       
%        \begin{itemize}
%         \item $\VSet^3=\ASet^2 \cup (\FSet_2 \setminus 
% 	\YSet^2)=\{\qSym[3],\qSym[4],\qSym[6]\}$,
%         \item $\ASet^3=\VSet^2 \cup (\FSet_2 \cap \YSet^2) = \{\qSym[1]\}$.
%        \end{itemize}       
%        At the fourth step:       
%        \begin{itemize}
%         \item $\VSet^4=\ASet^3 \cup (\FSet_3 \setminus 
% \YSet^3)=\{\qSym[0],\qSym[1]\}$,
%         \item $\ASet^4=\VSet^3 \cup (\FSet_3 \cap \YSet^3) = 
% \{\qSym[3],\qSym[4],\qSym[6]\}$.
%        \end{itemize}       
%        At the fifth step:       
%        \begin{itemize}
%         \item $\VSet^5=\ASet^4 \cup (\FSet_4 \setminus 
% \YSet^4)=\{\qSym[3],\qSym[4],\qSym[6]\}$,
%         \item $\ASet^5=\VSet^4 \cup (\FSet_4 \cap \YSet^4) = 
% \{\qSym[0],\qSym[1]\}$.
%        \end{itemize}       
%        Finally, at the sixth step:       
%        \begin{itemize}
%         \item $\VSet^6=\ASet^5 \cup (\FSet_5 \setminus 
% 	\YSet^5)=\{\qSym[0],\qSym[1],\qSym[2],\qSym[5]\}$,
%         \item $\ASet^6=\VSet^5 \cup (\FSet_5 \cap \YSet^5) = 
% \{\qSym[3],\qSym[4],\qSym[6]\}$.
%        \end{itemize}
       
%\end{itemize}
% 
% 	\noindent
	At this point we have that $\sforce(\VSet^6) = \{\qSym[0],\qSym[1],\qSym[5],
	\qSym[6]\} 	\neq \emptyset = \YSet[][5]$.
	This means that the fixpoint for $\YSet[][5]$ has not been reached yet.
	Then, we update the set $\YSet[][5]$ with the new value and compute again
	$\VSet[][6]$.
	This procedure is repeated up to the point in which $\sforce(\VSet^6) =
	\YSet[][5]$, which means that the fixpoint for $\YSet[][5]$ has been reached.
	Then we iteratively proceed to compute $\YSet[][4] = \PsSet \setminus
	\YSet[][5]$ until a fixpoint for $\YSet[][4]$ is reached.
	Note that the sets $\ASet[][i]$ and $\VSet[][i]$ depends on the $\YSet[][i]$
	and so they need to be updated step by step.
	As soon as a fixpoint for $\YSet[1]$ is reached, the algorithm returns the set
	$\PsSet \setminus \YSet[1]$.
	As a fundamental observation, note that, due to the fact that the fixpoint
	operations are nested one to the next, updating the value of $\YSet[][i]$
	implies that every $\YSet[][j]$, with $j>i$, needs to be reset to the
	empty set.
	
	We now prove the correctness of this procedure.
	Note that the algorithm is an adaptation of the one provided by Kupferman and
	Vardi in~\cite{KV98}, for which a proof of correctness has never been shown.
	
	\begin{theorem}
		Let $\GName = \ArStr[EPG]$ be an \EPG with $\alpha$ being the parity
		sequence condition.
		Then, the following properties hold.
		
		\begin{enumerate}
			\item
				If $\alpha= \varepsilon$ then $\WinSet[0](\GName) =
				\WinSet[0](\alpha, \VSet, \ASet)$ and $\WinSet[1](\GName) =
				\WinSet[1](\alpha, \VSet, \ASet)$;
				
			\item
				If $\alpha$ starts with an odd parity set, it holds that
				$\WinSet[0](\GName) = \WinSet[0](\alpha, \VSet, \ASet)$;
				
			\item
				If $\alpha$ starts with an even parity set, it holds that
				$\WinSet[1](\GName) = \WinSet[1](\alpha, \VSet, \ASet)$.
		\end{enumerate}

	\end{theorem}
	
	\begin{proof}
		
		The proof of Item~1 follows immediately by definition, as $\alpha =
		\epsilon$ forces the two players to reach their respective winning sets in
		one step.
		
		For Item~2 and ~3, we need to find a partition of $\FSet$ into a winning set
		for Player $\frsentnam$ and a winning set for Player $\scnentnam$ such that
		the game is invariant \wrt the winning sets, once they are moved to visiting
		and avoiding, respectively.
		We proceed by mutual induction on the length of the sequence $\alpha$.
		As base case, assume $\alpha = \FSet$ and $\FSet$ to be an odd parity set.
		Then, first observe that Player $\frsentnam$ can win only by eventually
		hitting the set $\VSet$, as the parity condition is made by only odd
		numbers.
		We have that $\WinSet[0](\FSet, \VSet, \ASet) = \mu \YSet
		(\PsSet \setminus \WinSet[1](\varepsilon, \ASet \cup (\FSet \setminus
		\YSet), \VSet \cup (\FSet \cap (\YSet)))) = \mu \YSet (\PsSet \setminus
		\sforce(\ASet \cup (\FSet \setminus \YSet)))$ that, by definition,
		computes the set from which Player $\scnentnam$ cannot avoid a visit to
		$\VSet$, hence the winning set for Player $\frsentnam$.
		In the case the set $\FSet$ is an even parity set the reasoning is
		symmetric.
		
		As an inductive step, assume that Items~2 and~3 hold for sequences $\alpha$
		of length $n$, we prove that it holds also for sequences of the form $\FSet
		\cdot \alpha$ of length $n + 1$.
		Suppose that $\FSet$ is a set of odd priority.
		Then, we have that, by induction hypothesis, the formula $\WinSet[1](\alpha,
		\ASet \cup (\FSet \setminus \YSet), \VSet \cup (\FSet \cap \YSet))$ computes
		the winning set for Player $\scnentnam$ for the game in which the
		nodes in $\FSet \cap \YSet$ are visiting, while the nodes in $\FSet
		\setminus \YSet$ are avoiding.
		Thus, its complement $\PsSet \setminus \WinSet[1](\alpha, \ASet \cup (\FSet
		\setminus \YSet), \VSet \cup (\FSet \cap \YSet))$ returns the winning set
		for Player $\frsentnam$ in the same game.
		Now, observe that, if a set $\YSet'$ is bigger than $\YSet$, then $\PsSet
		\setminus \WinSet[1](\alpha, \ASet \cup (\FSet \setminus \YSet'), \VSet
		\cup (\FSet \cap \YSet'))$ is the winning set for Player $\frsentnam$ in
		which some node in $\FSet \setminus \YSet$ has been moved from avoiding
		to visiting.
		Thus we have that $\PsSet \setminus \WinSet[1](\alpha, \ASet \cup
		(\FSet \setminus \YSet), \VSet \cup (\FSet \cap \YSet)) \subseteq \PsSet
		\setminus \WinSet[1](\alpha, \ASet \cup (\FSet \setminus \YSet'), \VSet
		\cup (\FSet \cap \YSet'))$.
		Moreover, observe that, if a node $\qElm \in \FSet \cup \ASet$ is winning for
		Player $\frsentnam$, then it can be avoided in all possible winning plays,
		and so it is winning also in the case $\qElm$ is only in $\FSet$.
		It is not hard to see that, after the last iteration of the fixpoint
		operator, the two sets $\FSet \setminus \YSet$ and $\FSet \cap \YSet$ can be
		considered in avoiding and winning, respectively, in a way that the winning
		sets of the game are invariant under this update, which concludes the proof
		of Item~2.
		
		Also in the inductive case, the reasoning for Item~3 is perfectly symmetric
		to the one for Item~2.
	\end{proof}

\end{section}

% End of file SectionI.tex

% 	\input{SectionII}

	% Begin of file SectionIII.tex

% \newpage

\begin{section}{Implementation of \APT\ in \PGSolver}
	\label{sec:imp}

\setlength{\columnsep}{-4pt}%

		\begin{wrapfigure}[19]{r}{0.55\textwidth}%
			\vspace{-4.8em}
% 			\hspace{-2em}
			\begin{flushleft}

\begin{lstlisting}[basicstyle={\scriptsize},belowskip={0.006cm},frame=none,
language=Python,mathescape=true,tabsize=1,showstringspaces=false]
			
			$\mathbf{fun}$ win $i$ $\mathcal{G}$ $\PsSet$ $\alpha$ $\VSet$ $\ASet$ =
			 if($\alpha \neq \epsilon)$ $\mathbf{then}$
			  W := $\PsSet\ \setminus$(min_fp (1-$i$) $\mathcal{G}$ $\PsSet$ $  \alpha\ 
			\ASet\ \VSet$);
			 else
				 W :=$\force_i( \VSet) $;
			 $\mathbf{return}$ W;;

			$\mathbf{fun}$ min_fp $i$ $\mathcal{G}$ $\PsSet$ $\alpha$ V A  =
			 $\YSet_1$ := $\emptyset $;
			 $\YSet_2$ := $\emptyset $;
			 $\FSet$ := head[$\alpha$];
			 $\alpha'$ := tail[$\alpha$]; 
			 $\VSet'$ := $\VSet \cup \FSet$;
			 $\ASet'$ := $\ASet$;
			 
			 $\YSet_2$ := win $i$ $\mathcal{G}$ $\PsSet$ $\alpha'$ $\VSet'$ $\ASet'$ ;

			 while( $\YSet_2 \neq \YSet_1$) $\mathbf{do}$ 
			  $\YSet_1$ := $\YSet_2$ ;
			  $\VSet'$ := $\VSet \cup (\FSet \cap \YSet_1)$;
			  $\ASet'$ := $\ASet \cup (\FSet \setminus \YSet_1$);
			  $\YSet_2$ := win $i$ $\mathcal{G}$ $\PsSet$ $\alpha'$  $\VSet'$ $\ASet'$ ;
			 $\mathbf{done}$	     
			 $\mathbf{return}$ $\YSet_2$;;

		
			\end{lstlisting}
			
			\end{flushleft}
			
			%	$\mathbf{let}$ force game v a i =			
		%		$\mathbf{let}$ set = ref(NdsSet.empty) in 
		%	$\mathbf{let}$ l = Array.length game in
				
			%	$\mathbf{for}$ j=0 to l-1 $\mathbf{do}$
		%	(  
			%	$\mathbf{let}$ acc = check game j v a i in
			%	$\mathbf{if}$(acc = true)$\mathbf{then}$
			%	set:= NdsSet.add (Node j) !set        
			%	)
			%	$\mathbf{done}$;             
				
			%	;; 
\vspace{-0.4cm}
\protect\caption{APT Algorithm}
\label{APT}\vspace{-2.3cm}
\end{wrapfigure}
		 
		In this section we describe the implementation of \APT\ in the well-known
		platform \PGSolver\ developed in OCaml by Friedman and Lange~\cite{FL09},
		which collects the large majority of the algorithms introduced in the
		literature to solve parity games~\cite{Zie98, Jur00, HKLN12, VJ00, Sch08,
		JPZ08, Sch07}.
		
		We briefly recall the main aspects of this platform.
		The graph data structure is represented as a fixed length array
		of tuples.
		Every tuple has all information that a node needs, such as the owner
		player, the assigned priority and the adjacency list of nodes.
		The platform implements a collection of tools to generate and solve parity
		games, as well as compare the performance of different algorithms.
		The purpose of this platform is not just that of making available an environment 
		to deploy and test a generic solution algorithm, but also to
		investigate the practical aspects of the different algorithms on the
		different classes of parity games.
		Moreover, \PGSolver\ implements optimizations that can be applied to all
		algorithms in order to improve their performance.
		The most useful optimizations in practice are decomposition into strongly
		connected components, removal of self-cycles on nodes, and priority
		compression.

		We have added to \PGSolver\ an implementation of the \APT\
		algorithm introduced in Section~\ref{sec:pargam}.
		Our procedure applies the fixpoint algorithm to compute the set of winning
		positions in the game by means of two principal functions that implement
		the two functions of the algorithm core processes, i.e.,  function
		$\force_{i}$ and the recursive function $\WinSet_{i}(\alpha,V,A)$.
		The pseudocode of the \APT\ algorithm implementation is reported in Figure
		\ref{APT}. It takes six parameters: the
		Player ($\frsentnam$ or~$\scnentnam$), the game, the set of
		nodes, the condition $\alpha$, the set of visiting and avoiding.
		Moreover, we define the function \emph{min\_fp} for the calculation
		of the fixed point.
		The whole procedure makes use of Set and List data structures, which are
		available in the OCaml's standard library, for the manipulation of the sets
		visiting and avoiding, and the accepting condition $\alpha$.
    The tool along with the implementation of the \APT\ algorithm is available
		for	download from \texttt{https://github.com/antoniodistasio/pgsolver-APT}.
%Our procedure partition the nodes of a parity game $\GmName$ into its related
%sets $\F_i$ with $\F_i = \priFun^{-1}(i)$, that is the set of all positions 
%labeled with $i$. Then, it costructs the parity condition $\alpha$ as a list 
% of  sets $\F_1,\cdots ,\F_k$.
	
For the sake of clarity, we report that in \PGSolver\ it is used the
maximal priority to decide who wins a given parity game. Conversely, the \APT\
algorithm uses the minimal priority. However, these two conditions are
well known to be equivalent and, in order to compare instances of the same game
on different implementations of parity games algorithms in \PGSolver, we simply
convert the game to the specific algorithm accordingly. For the conversion, we
simply use a suitable permutation of the priorities.

% 	Finally, it is worth noting that an instance $\GmName$ of a
% parity game is
% 	differently interpreted by \APT, on one side, and the other solvers
% 	implemented in \PGSolver, on the other side.
% 	Indeed, the first interprets the parity condition as the minimum priority
% 	occurring infinitely often, while the others interpret it for the maximum one.
% 	Then, first we transform a game $\GmName$ into a one $\GmName'$ by means of a 
% 	suitable permutation of priorities, and then we apply the procedure.
\end{section}

% End of file SectionIII.tex

	% Begin of file SectionIV.tex

% \newpage

\begin{section}{Experiments}
	\label{sec:exp}

	In this section, we report the experimental results on evaluating the performance 
  for the \APT\ algorithm implemented in \PGSolver\ over the random benchmarks
	generated in the platform. We have compared the performance of the
	implementation of \APT\	with those of \RE\ and \SP. We have chosen
	these two algorithms as they have been proved to be the best-performing in
	practice~\cite{FL09}.
	
	All tests have been run on an AMD Opteron 6308 @2.40GHz, with 224GB of RAM and
	128GB of swap running Ubuntu 14.04.
	We note that \APT\ has been executed without applying any optimization
	implemented in \PGSolver~\cite{FL09}, while \SP\ and \RE\ are run with such
	optimizations.
	Applying these optimization on \APT\ is a topic of further research.
	
	We evaluated the performance of the three algorithms over a set of games that
	are randomly generated by \PGSolver, in which it is possible to give the
	number $n$ of states and the number $k$ of priority as parameters.
	We have taken $20$ different game instances for each set of parameters and
	used the average time among them returned by the tool.
	For each game, the generator works as follows.
	For each node $\psElm$ in the graph-game, the priority $\priFun(\psElm)$ is
	chosen uniformly between $0$ and $k - 1$, while its ownership is assigned to
	Player $\frsentnam$ with probability $\frac{1}{2}$, and to Player $\scnentnam$
	with probability $\frac{1}{2}$.
	Then, for each node $\psElm$, a number $d$ from $1$ to $n$ is chosen uniformly
	and $d$ distinct successors of $\psElm$ are randomly selected.

	\begin{subsection}{Experimental results}
		\label{sec:exp;sub:res}
		
		\begin{table}

			\begin{center}

			\scalebox{0.80}[0.80]{
				\begin{tabular}{|c||c|c|c|c|c|c|c|c|c|c|c|c}

					\hline

					{\small{}}  & \multicolumn{3}{ |c| }{2 Pr}  & \multicolumn{3}{|c|}{3
					Pr} & \multicolumn{3}{ |c| }{5 Pr} \\

					{\small{}$n$}  & {\small{}\RE} & {\small{}\SP} &  {\small{}\APT} & 
					{\small{}\RE} & {\small{}\SP} &  {\small{}\APT} & 
					{\small{}\RE} & {\small{}\SP} &  {\small{}\APT}  \\
					\hline
					\hline
					{\scriptsize{}2000} &   
					{\scriptsize{}4.94} & {\scriptsize{}5.05} & {\scriptsize{}0.10} &
					{\scriptsize{}4.85} & {\scriptsize{}5.20} & {\scriptsize{}0.15} &
					{\scriptsize{}4.47} & {\scriptsize{}4.75} &
					{\scriptsize{}0.42} \\
					\hline
					{\scriptsize{}4000} &  
					{\scriptsize{}31.91} & {\scriptsize{}32.92} & {\scriptsize{}0.17} &
					{\scriptsize{}31.63} & {\scriptsize{}31.74} & {\scriptsize{}0.22} &
					{\scriptsize{}31.13} & {\scriptsize{}32.02} & {\scriptsize{}0.82} 
					\\
					\hline
					{\scriptsize{}6000} &  
					{\scriptsize{}107.06} & {\scriptsize{}108.67} & {\scriptsize{}0.29} &
					{\scriptsize{}100.61} & {\scriptsize{}102.87} & {\scriptsize{}0.35} &
					{\scriptsize{}100.81} & {\scriptsize{}101.04} & {\scriptsize{}1.39} 
					\\
					\hline
					{\scriptsize{}8000} & 
					{\scriptsize{}229.70} & {\scriptsize{}239.83} & {\scriptsize{}0.44} &
					{\scriptsize{}242.24} & {\scriptsize{}253.16} & {\scriptsize{}0.5}  &
					{\scriptsize{}228.48} & {\scriptsize{}245.24} & {\scriptsize{}2.73} 
					\\
					\hline
					{\scriptsize{}10000} & 
					{\scriptsize{}429.24} & {\scriptsize{}443.42} & {\scriptsize{}0.61} &
					{\scriptsize{}482.27} & {\scriptsize{}501.20} & {\scriptsize{}0.85} &
					{\scriptsize{}449.26} & {\scriptsize{}464.36} & {\scriptsize{}3.61} 
					\\
					\hline
					{\scriptsize{}12000} &  
					{\scriptsize{}772.60} & {\scriptsize{}773.76} & {\scriptsize{}0.87} &
					{\scriptsize{}797.07} & {\scriptsize{}808.96} & {\scriptsize{}0.98} &
					{\scriptsize{}762.89} & {\scriptsize{}782.53} & {\scriptsize{}6.81} 
					\\
					\hline
					{\scriptsize{}14000} &
					{\scriptsize{}1185.81} & {\scriptsize{}1242.56} & {\scriptsize{}1.09}&
					{\scriptsize{}1227.34} & {\scriptsize{}1245.39} & {\scriptsize{}1.15}&
					{\scriptsize{}1256.32} & {\scriptsize{}1292.80} &
					{\scriptsize{}10.02} \\
					\hline

				\end{tabular}
				}

			\vspace{0.2cm}

			\caption{Runtime executions with fixed priorities $2$, $3$ and $5$}
			\label{tab2}
			\end{center}
			
			\vspace{-4em}

		\end{table}

		We ran two experiments.
		First, we tested games with $2$, $3$, and $5$ priorities, where for each of
		them we measured runtime performance for different state-space sizes,
		ranging in $\{ 2000, 4000, 6000, 8000 , 10000, 12000, 14000 \}$.
		The results are in Table~\ref{tab2}, in which the number of states is
		reported in column~1, the number of colors is reported in the
		macro-column~2, 3, and 5, each of them containing the runtime executions,
		expressed in seconds, for the three algorithms.
		%timeout limit is one hour
		\begin{wraptable}[23]{r}{0.39\textwidth}
			% 			\begin{table}
			
			\vspace{-3em}
			
			\begin{center}
				
				\scalebox{0.80}[0.80]{
				\begin{tabular}{c||c|c|c|c}
					\hline
					
					\multicolumn{1}{c}{{\small $n$}} & \multicolumn{1}{c}{{\small Pr}} & 
					\multicolumn{1}{c}{{\RE}} & \multicolumn{1}{c}{{\SP}} & {\APT} \\
					\multicolumn{5}{c}{} \\[-1em]
					\multicolumn{5}{c}{\small $n = 2^{k}$} \\
					\hline
					\hline
					{\scriptsize{}1024} & {\scriptsize{}10} &  
					{\scriptsize{}1.25} & {\scriptsize{}1.25} & {\scriptsize{}8.58}\\
					\hline
					{\scriptsize{}2048} & {\scriptsize{}11} & 
					{\scriptsize{}7.90} & {\scriptsize{}8.21} & {\scriptsize{}71.08}\\
					\hline
					{\scriptsize{}4096} & {\scriptsize{}12} & 
					{\scriptsize{}52.29} & {\scriptsize{}52.32} & 
					{\scriptsize{}1505.75}\\
					\hline
					{\scriptsize{}8192} & {\scriptsize{}13} & 
					{\scriptsize{}359.29} & {\scriptsize{}372.16} & 
					{\scriptsize{}\abortT}\\
					\hline
					{\scriptsize{}16384} & {\scriptsize{}14} & 
					{\scriptsize{}2605.04} & {\scriptsize{}2609.29} & 
					{\scriptsize{}\abortT}\\
					\hline
					{\scriptsize{}32768} & {\scriptsize{}15} & 
					{\scriptsize{}\abortT} & {\scriptsize{}\abortT} & 
					{\scriptsize{}\abortT}\\
					\hline					
					\hline
					
					\multicolumn{5}{c}{\small $n = e^{k}$} \\
					\hline
					\hline
					% \scriptsize{7} & \scriptsize{2} &  
					% \scriptsize{0} & \scriptsize{0} & \scriptsize{0} \\
					% \hline
					\scriptsize{21} & \scriptsize{3} & 
					\scriptsize{0} & \scriptsize{0} & \scriptsize{0} \\
					\hline
					\scriptsize{55} & \scriptsize{4} & 
					\scriptsize{0} & \scriptsize{0} & \scriptsize{0.02} \\
					\hline
					\scriptsize{149} & \scriptsize{5} & 
					\scriptsize{0.01} & \scriptsize{0.01} & \scriptsize{0.08}  \\
					\hline
					\scriptsize{404} & \scriptsize{6} & 
					\scriptsize{0.14} & \scriptsize{0.14} & \scriptsize{0.19}\\
					\hline
					\scriptsize{1097} & \scriptsize{7} & 
					\scriptsize{1.72} & \scriptsize{1.72} & \scriptsize{0.62} \\
					\hline
					\scriptsize{2981} & \scriptsize{8} & 
					\scriptsize{24.71} & \scriptsize{24.46} & \scriptsize{7.88}\\
					\hline
					\scriptsize{8104} & \scriptsize{9} & 
					\scriptsize{413.2.34} & \scriptsize{414.65} & \scriptsize{35.78}\\
					\hline
					\scriptsize{22027} & \scriptsize{10} & 
					\scriptsize{\abortT} & \scriptsize{\abortT} & \scriptsize{311.87}\\

					\hline
					\multicolumn{5}{c}{} \\[-1em]
					\multicolumn{5}{c}{\small $n = 10^{k}$} \\
					\hline
					\hline
					{\scriptsize{}10} & {\scriptsize{}1} &  
					{\scriptsize{}0} & {\scriptsize{}0} & {\scriptsize{}0}\\
					\hline
					{\scriptsize{}100} & {\scriptsize{}2} & 
					{\scriptsize{}0} & {\scriptsize{}0} & {\scriptsize{}0}\\
					\hline
					{\scriptsize{}1000} & {\scriptsize{}3} & 
					{\scriptsize{}1.3} & {\scriptsize{}1.3} & {\scriptsize{}0.04}\\
					\hline
					{\scriptsize{}10000} & {\scriptsize{}4} & 
					{\scriptsize{}738.86} & {\scriptsize{}718.24} & 
					{\scriptsize{}4.91}\\
					\hline
					{\scriptsize{}100000} & {\scriptsize{}5} & 
					{\scriptsize{}\abortM} & {\scriptsize{}\abortM} & 
					{\scriptsize{}66.4}\\
					\hline

					\hline
					
				\end{tabular}
				}
				
				\vspace{-0.2em}
				\caption{Runtime executions with $n=e^k$ and $n=2^k$ and $n=10^k$}
				\label{table}
			\end{center}
			
			\vspace{-0.8cm}
			
			% 			\end{table}
		\end{wraptable}
		Second, we evaluated the algorithms on games with an exponential
	number of nodes \wrt the number of priorities.
	More precisely, we ran experiments for $n=2^k$, $n=e^k$ and $n=10^k$, where
	$n$ is the number of states and $k$ is the number of priorities.
	%MYV: what was our timeout limit?	
	
		The experiment results are reported in Table~\ref{table}.
		By \abortT, we denote that the execution has been aborted due to time-out
		(greater of one hour),
		while by \abortM we denote that the execution has been aborted due to
		mem-out.
	
		The first experiment shows that with a fixed number of priorities ($2$, $3$,
		and $5$) \APT\ significantly outperforms the other algorithms, showing
		excellent runtime execution even on fairly large instances.
		For example, for $n = 14000$, the running time for both \RE\ and \SP\ is
		about $20$ minutes, while for \APT\ it is less than a minute.

		The results of the exponential-scaling experiments, shown in
		Table~\ref{table}, give more nuanced results.
		Here, \APT\  is the best performing algorithm for $n = e^k$ and $n = 10^k$.
		For example, when $n=100000$ and $k=5$, both \RE\ and \SP\ memout, while
		\APT\ completes in just over one minute.
		That is, the efficiency of \APT\ is notable also in terms of memory usage.
		At the same \APT\ underperforms for $n=2^k$. Our conclusion is that \APT\
		has superior performance when the number of priorities is logarithmic in the
		number of game-graph nodes, but the base of the logarithm has to be large
		enough.
		As we see experimentally, $e$ is sufficiently large base, but $2$ is not.
		This point deserve further study, which we leave to future work.
		%As a final observation from the benchmarks reported in Table~\ref{tab2} and
		%Table~\ref{table}, observe that \APT\ performs much better than
		%others in case the number of priorities is logarithmic \wrt the number of nodes.
		In Figure~\ref{plot} we just report graphically the benchmarks in the case
		$n = e^k$.
		An interested reader can find more detailed experiment results at
		\texttt{https://github.com/antoniodistasio/pgsolver-APT}.
		
	\begin{figure}[!h]
 		\centering
 		\vspace{-1em}
 		\includegraphics[scale=0.30]{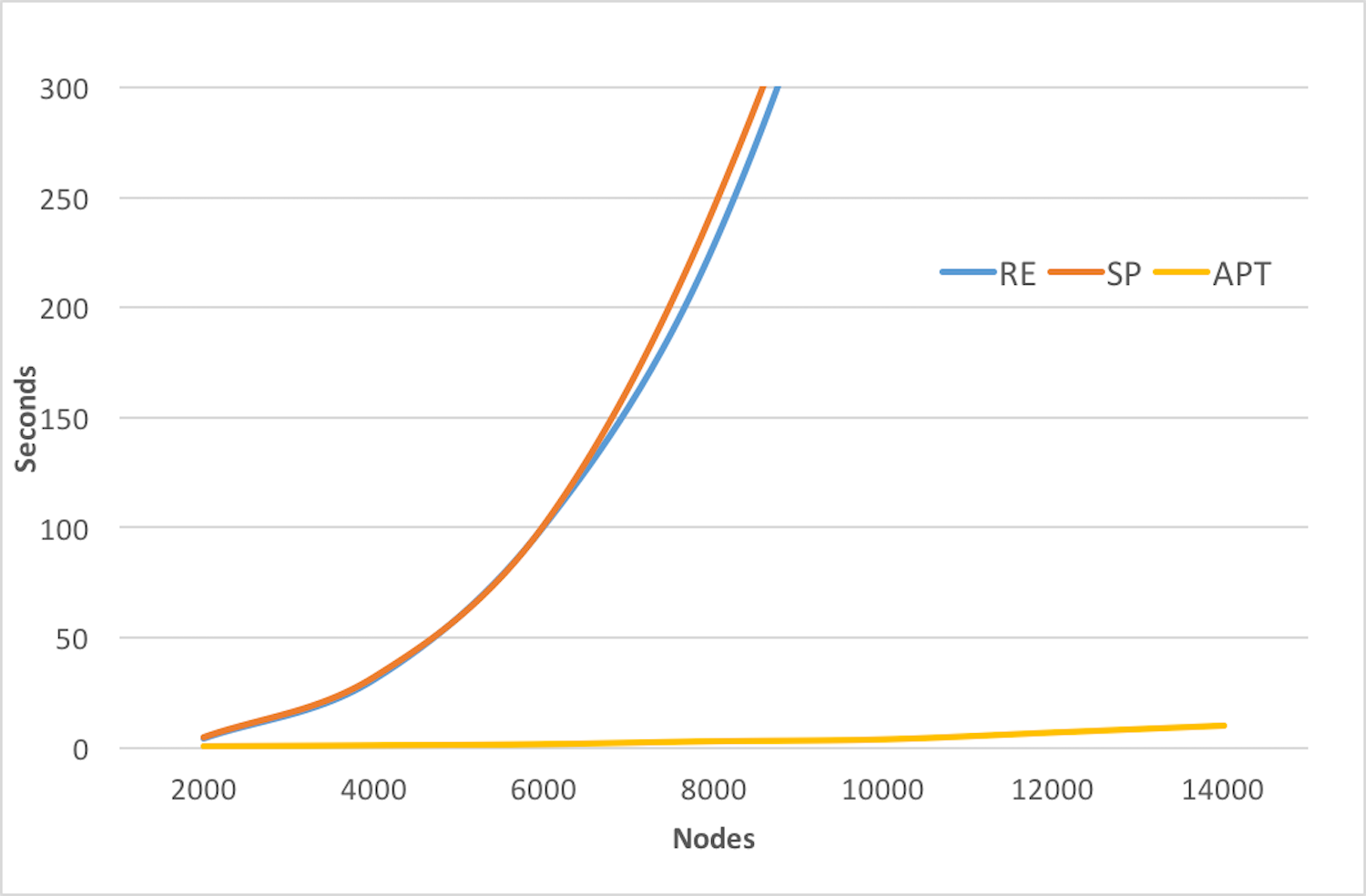}
 		\vspace{-1em}
 		\caption{Runtime executions with $n = e^k$}
 		\vspace{-3em}
 		\label{plot}
  
 	\end{figure}

\end{subsection}
\end{section}

% End of file SectionIV.tex

	% Begin of file Conclusion.tex

\begin{section}{Conclusion}

	The \APT\ algorithm, an automata-theoretic technique to solve parity games,
	has been designed two decades ago by Kupferman and Vardi~\cite{KV98}, but
	never considered to be useful in practice~\cite{FL09PGSolver}.
	In this paper, for the first time, we fill missing gaps and implement this
	algorithm.
	By means of benchmarks based on random games, we show that it is the best
	performing algorithm for solving parity games when the number of priorities is
	very small \wrt the number of states.
	We believe that this is a significant result as several applications of parity
	games to formal verification and synthesis do yield games with a very small
	number of priorities. 
	\\ \indent
	The specific setting of a small number of priorities opens up opportunities
	for specialized optimization technique, which we aim to investigate in future
	work.
	This is closely related to the issue of accelerated algorithms for 
        three-color parity games~\cite{DF07}. We also plan to study why the 
performance of the 
	\APT\ algorithm is so sensitive to the relative number of priorities, as 
shown in 
	Table~\ref{table}.

%         As a future work we also plan to investigate the 
%         reason that makes our algorithm less performing when the priorities are not small and 
% 	investigate how to improve it in this respect. 
\end{section}

% End of file Conclusion.tex

	\footnotesize
	\bibliographystyle{plain}
	\bibliography{References}

\end{document}